\documentclass[journal]{IEEEtran}

\usepackage{graphicx}
\usepackage{epstopdf}
\usepackage[cmex10]{amsmath}
\usepackage{amsfonts}
\usepackage[ruled,linesnumbered]{algorithm2e}
\usepackage{algorithmic}
\usepackage{array}
\usepackage{eqparbox}
\usepackage{epsfig}
\usepackage{color}
\usepackage{graphicx}
\usepackage{times}
\usepackage{float}
\usepackage{stfloats}
\usepackage{graphicx}
\usepackage{subfigure}
\usepackage{type1cm}
\usepackage{url}
\usepackage{cite}
\usepackage{enumerate}
\usepackage{lineno}
\usepackage{amsthm}
\usepackage{amssymb}
\usepackage{bbm}
\usepackage{bm}
\usepackage{multirow}
\usepackage{mathtools}

\usepackage[short]{optidef}
\allowdisplaybreaks 

\usepackage{enumerate}

\SetKwInput{KwInput}{Input}
\SetKwInput{KwOutput}{Output}

\newtheorem{Lemma}{Lemma}

\newtheorem{Definition}{Definition}

\theoremstyle{remark}
\newtheorem{remark}{Remark}

\begin{document}

\title{Rateless Coded Blockchain for Dynamic IoT Networks}

\author{Changlin Yang, \IEEEmembership{Member, IEEE}, Alexei Ashikhmin, \IEEEmembership{Fellow, IEEE}, Xiaodong Wang, \IEEEmembership{Fellow, IEEE}, Zibin Zheng, \IEEEmembership{Fellow, IEEE}
%

\thanks{C. Yang and Z. Zheng are with the School of Software Engineering, Sun Yat-Sen University, Zhuhai, Guangdong, China, 519000. (e-mail: \{yangchlin6, zhzibin\}@mail.sysu.edu.cn).

A. Ashikhmin is with the Communications and Statistical Sciences Research Department, Nokia Bell Laboratories, Murray Hill, NJ, USA (e-mail: alexei.ashikhmin@nokia-bell-labs.com)

X. Wang is with the Electrical Engineering Department, Columbia University, New York, NY, USA (e-mail: wangx@ee.columbia.edu).


}}

\maketitle
\newtheorem{proposition}{Proposition}

\begin{abstract}

A key constraint that limits the implementation of blockchain in Internet of Things (IoT) is its large storage requirement resulting from the fact that each blockchain node has to store the entire blockchain. This increases the burden on blockchain nodes, and increases the communication overhead for new nodes joining the network since they have to copy the entire blockchain. In order to reduce storage requirements without compromising on system security and integrity, coded blockchains, based on error correcting codes with fixed rates and lengths, have been recently proposed. This approach, however, does not fit well with dynamic IoT networks in which nodes actively leave and join. In such dynamic blockchains, the existing coded blockchain approaches lead to high communication overheads for new joining nodes and may have high decoding failure probability. This paper proposes a rateless coded blockchain with coding parameters adjusted to network conditions. Our goals are to minimize both the storage requirement at each blockchain node and the communication overhead for each new joining node, subject to a target decoding failure probability. We evaluate the proposed scheme in the context of real-world Bitcoin blockchain and show that both storage and communication overhead are reduced by 99.6\% with a maximum $10^{-12}$ decoding failure probability.


%
\end{abstract}


\maketitle

\section{Introduction}
\label{Intro}

The blockchain technology has found wide applications in 
%
Internet-of-things (IoT) \cite{dai2019blockchain}
due to its distributive, secure and integrated properties \cite{karame2016bitcoin}. 
%
One of the key challenges that restricts the implementation of blockchain in IoT devices is the high storage requirements for their users \cite{fan2022scalable}. 
For example, a Bitcoin blockchain needs 380 GB hard-disk space in 2022 \cite{bitcoin_size}  and takes several days to download for a fully functional Bitcoin Core client \cite{bitcoin_core}.
This makes it hard for portable devices, such as smart phones and laptops, to serve as full nodes in a blockchain, which weakens the distributive property. 
%

There are several ways to reduce the the blockchain storage requirements \cite{zhou2020solutions}. 
The first is simplified payment verification (SPV) client, namely {\em light} node \cite{nakamoto2008bitcoin}. 
A light node stores only the block headers and relies on the full nodes to verify transactions. 
Hence, it can be compromised if it relies on some malicious nodes \cite{karame2016bitcoin}. 
The second way is to prune old information that are not needed for verifying new blocks, i.e., the spent Unspent Transaction Output (UTXO) can be deleted when using Bitcoin \cite{bitcoin_core}.
This reduces the needed storage space while retaining the ability to `mine' new blocks. 
The work in \cite{dryja2019utreexo} further reduces the storage requirement of a pruned node by only storing the proofs of UTXOs it holds.
%
%
However, the deleted information will eventually be inaccessible if all nodes are pruned node. 
%
%
%
The last approach to reduce storage is `sharding' \cite{zamani2018rapidchain} \cite{huang2022brokerchain} the blockchain nodes into groups. Each group of nodes stores and operates a portion of the entire blockchain and so the storage at each node is reduced. However, the security level of the `sharding' approach is directly related to the group size \cite{das2018security}. 
Hence, all these solutions compromise on blockchain security or/and integrity. 

Recently, coded blockchains have been proposed to reduce the blockchain storage requirements without loss of security and integrity \cite{dai2018low}\cite{perard2018erasure}\cite{wu2020distributed} \cite{raman2021ToN} \cite{yang2021storage}. 
They can also improve block propagation efficiency \cite{cebe2018network} \cite{chawla2019velocity} \cite{9609913} \cite{zhang2022speeding}.
In addition to storage and transmission cost reduction, the coded blockchains can improve the security of existing SPV client via fraud proofs \cite{al2018fraud} \cite{santini2022optimization} or  coded merkle trees \cite{yu2020coded} \cite{mitra2021overcoming} \cite{mitra2022polar}, and also the security of sharding solutions via coded shards \cite{li2020polyshard} \cite{sasidharan2021private}  \cite{gadiraju2020secure}, etc.
In such blockchains, nodes encode the {\em data} of blocks using fixed rate error correction codes, such as Reed-Solomon codes  \cite{1994Reed}, low-density parity-check (LDPC) codes \cite{2002Low} and polar codes \cite{PolarCodes}, and distributively store the encoded blocks. 
In particular, a group of $k$ blocks are encoded into $n$ encoded blocks using an $[n, k]$ code. Then the storage at nodes reduces to about $1/k$ that of the replicated blockchain.
When the {\em data} of a particular block is needed by a node, it can be recovered by collecting slightly more than $k$ encoded blocks from other nodes and running a decoder of the error correction code.
Nodes can encode only old blocks that are not needed for `mining', which resembles the old block pruning approach, but without any loss of security and integrity. 
%
%
%
%

In practice, an IoT blockchain network is dynamic and nodes frequently join or leave the network. 
%
%
Although for most of the time, the numbers of nodes joining and leaving are steady, sharp changes may occur sometimes. For example, the number of Bitcoin nodes dropped by 24\% in Sep. 2020 \cite{bitcoin_number_2015}.
The encoding schemes in \cite{dai2018low} -\cite{gadiraju2020secure} are designed for networks with fixed  number of nodes, but some nodes can be temporarily unavailable.  These encoding schemes are not well-suited for  dynamic networks with the following reasons.
%


\begin{itemize}
    
     \item In a dynamic network, the number of nodes may rapidly change. Thus using an $[n,k]$ code of a fixed length $n$ becomes very inconvenient.  In contrast, rateless codes
     is a natural solution for this scenario.

    \item The nodes in a dynamic network often permanently leave the network. Thus, if we use a fixed length $[n,k]$ code, but many of the $n$ nodes, that were available at the time of encoding, permanently leave, then other nodes will not be able to reconstruct these $k$ blocks.

    \item One possible way of using fixed length codes in a dynamic network is to use an $[N, k]$ code with length $N$ significantly larger than the number of network nodes $n$. So, when a new node, say node $n+1$, joins the network it can compute its own parity check symbol. However, in this approach, the new node first has to download more than $k$ encoded blocks from other nodes to decode the $k$ original blocks, which means very large communication overhead and high decoding complexity. This problem does not exist in the proposed rateless coded blockchain. 
    

\end{itemize}

Rateless Luby transform (LT) codes \cite{luby2002lt} were used to encode the blockchains in \cite{kadhe2019sef} \cite{quan2019transparent} \cite{pal2020fountain} \cite{qin2022downsampling}. Moreover, an attempt of raptor code is discussed in \cite{tiwari2021secure}. 
However, the parameter $k$ is fixed in these works and as a result these schemes do not well fit dynamical blockchains. In addition, most of these works use non-systematic LT codes. This forces a new joining node first to decode all previously generated codewords. The number of these codewords can be very large since they contain almost entire blockchain and decoding of even a single codeword 
takes visible computational resources. Thus, this approach results in very large computational and communication resources needed for new nodes.

In this work, we use rateless raptor codes \cite{shokrollahi2006raptor}. These codes have a low decoding failure probability than LT codes \cite{shokrollahi2006raptor}.
%
%
We use systematic rateless codes such that the original blocks are also stored at some nodes in the network. Thus, when a node joins the network, it only needs to collect a few original blocks and encode its own parity check (or original) blocks for each group of previously encoded blocks.
%
We adjust parameter $k$ based on the network condition, which allows us to find a good trade-off between storage efficiency and blockchain reliability.  
We propose to use novel {\em enhanced blocks} to store the coding parameters of each group of encoded blocks. 
%
%
We evaluate the proposed rateless coded blockchain in the Bitcoin network. The results 
show that by ensuring a less than $10^{-12}$ decoding failure probability,  the node storage and the communication overheads needed by a new joining node both reduce to 0.4\% as compared with the replicated blockchain.

This paper addresses a critical challenge that has not been considered by previous coded blockchain studies: how to ensure successful decoding when the number of nodes in the network changes. In particular, all existing coded blockchain works use a pre-determined information symbol length $k$. 
If sufficiently many nodes involved in encoding a particular codeword eventually leave the network, and this is a very possible event, it will be impossible to recover these $k$ information symbols. As a result, many blocks will be lost, which is not acceptable in most blockchain applications.  
%
Existing works do not have mechanism to update $k$. In this respect, this paper proposes a novel method to update $k$, encoding procedure for new joining nodes, an enhanced block structure to ensure the coding parameter achieves consensus among all nodes in the network. Up to the best of our knowledge all these important questions have not been addressed in the literature yet.

The remainder of this paper is organized as follows. 
Section \ref{sec_background} presents the background on coded blockchains and rateless codes. Section \ref{sec_coding_scheme} gives an overview of the proposed rateless coded blockchain. Section \ref{sec_code_block} presents the structure of our novel enhanced blocks. 
Section \ref{sec_NMA} and \ref{sec_k} present the protocol for a new node joining the network and our algorithm for choosing parameter $k$, respectively.  Simulation results are given in Section \ref{sec_simulation}. Finally, Section \ref{sec_conclusion} concludes the paper.


\section{Background on Coded Blockchain and Rateless Code \label{sec_background} }

\subsection{Background on Coded Blockchain}

In this subsection, we briefly describe how coded blockchains are organized in \cite{perard2018erasure} and \cite{wu2020distributed}. We will generalize this approach for the case of using rateless codes later in our work. 

Let $W$ be the index of the last block, and $B^W$ be the data of this block (that is the $W$-th block without its header). 
Slightly abusing notation, in what follows we will be calling by block both entire block  (data + header) and its data part $B^w$.
Let $\alpha$ be a sufficiently large integer such that with high probability blocks $0, 1, \ldots, W-\alpha$ do not contain forks that would grow. These blocks are called {\em confirmed blocks}. In a coded blockchain the confirmed blocks are partitioned into $M$ groups, each of which has been encoded by the fixed rate code, as it is shown in Fig.~\ref{block_groups}. 
%
%
%
Each group consists of $k$ blocks and each block $B^w$ is divided into $s$ symbols over the finite  field $\mathbb{F}_q$, denoted by $\mathbf{b}_w$. Assuming that the $m$-th group has blocks with indices $\{w_1, w_2, \ldots, w_k \}$, this group can be represented by a matrix with $k$ columns, denoted by $\mathbf{b}_{w_i}$,  and each column forms the corresponding original block of the group: 
%
\begin{align}
\label{equ_original}
\begin{bmatrix} 
    b_{w_1, 1} & b_{w_2, 1} & \ldots & b_{w_k, 1} \\
    b_{w_1, 2} & b_{w_2, 2} & \ldots & b_{w_k, 2} \\
    \vdots &  \vdots &  \vdots &  \vdots \\
    b_{w_1, s} & b_{w_2, s} & \ldots & b_{w_k, s} 
    \end{bmatrix}.
\end{align}
The $k$ entries of each row of \eqref{equ_original} are information symbols.  So, each row of \eqref{equ_original} is encoded into $n$ coded symbols using a $[n, k]_q$ linear code over $\mathbb{F}_q$, where  $q = 2^p$ and $p$ is the number of bits in each symbol $b_{w_i, j}$, with a $k \times n$ generator matrix $\mathbf{G}$. 
%
%
The code length $n$ is typically equal to the number of nodes in the network at the moment of encoding.
The encoding is conducted by multiplication of each row of \eqref{equ_original} by $\mathbf{G}$. As a result we obtain the vector 
\begin{align}
\label{cal_inter_blocks}
[u_{1, h}, \ldots, u_{n, h} ] = [ b_{w_1, h}, \ldots b_{w_k, h} ] \mathbf{G}.
\end{align}
Thus after encoding all rows of \eqref{equ_original}, we obtain the matrix
\begin{align}
\label{equ_intermedia}
\begin{bmatrix} 
    u_{1, 1} & u_{2, 1} & \ldots & u_{k, 1} & \ldots & u_{n, 1} \\
    u_{1, 2} & u_{2, 2} & \ldots & u_{k, 2} & \ldots & u_{n, 2} \\
    \vdots &  \vdots &  \vdots &  \vdots  &  \vdots &  \vdots \\
    u_{1, s} & u_{2, s} & \ldots & u_{k, s} & \ldots & u_{n, s}
    \end{bmatrix}.
\end{align}
Each node stores only one column of \eqref{equ_intermedia}, which greatly reduces the needed storage. We will denote the $i$-th column of \eqref{equ_intermedia} by $\mathbf{u}_i$. 

\begin{figure}[t]
\centering
   \includegraphics[scale = 0.5]{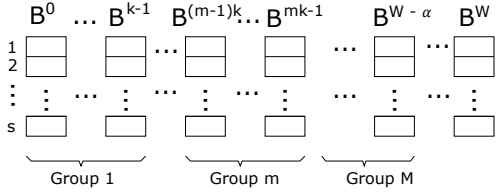}
  \caption{Blocks are encoded by groups. Each group consists of $k$ blocks. Each block is divided into $s$ symbols. }
  \label{block_groups}
\end{figure}


It is convenient to use systematic $\mathbf{G}$, such that the first $k$ columns of $\mathbf{G}$ form the identity matrix $\mathbf{I}_k$. In this case we have $\mathbf{u}_i = \mathbf{b}_{w_i}, i = 1, ...,k$. 
If a node, say node $j$, needs block $\mathbf{u}_i$, for example for verifying a transaction, it broadcasts a request and obtains $\mathbf{u}_i$. If the node that stores $\mathbf{u}_i$ is not available, then node $j$ collects a few other columns of \eqref{equ_intermedia} and use a decoder of the $[n, k]_q$ code to repair $\mathbf{u}_i$ \cite{wu2020distributed}. 
%
On the other hand, when repairing in \cite{wu2020distributed} fails or $\mathbf{G}$ is a non-systematic generator matrix, as in \cite{perard2018erasure}, node $j$ needs to collect $(1+\epsilon)k$ columns of \eqref{equ_intermedia} from other nodes and to decode all the information symbols $\mathbf{b}_{w_1} = \mathbf{u}_1, ..., \mathbf{b}_{w_k} = \mathbf{u}_k$.
Here $\epsilon$ is the {\em decoding overhead} that varies for different codes.

\subsection{Background on Rateless Codes}

We now briefly review the rateless codes and their variances. 
By using the rateless codes, 
the transmitter uses $n$ information symbols to generate potentially unlimited number of coded symbols, and
sends them to one or more receivers.
As soon as a receiver collects any $(1+\epsilon)n$ coded symbols, it can decode, with a high probability of success, the $n$ information symbols. 

\subsubsection{LT code \cite{luby2002lt} \label{sec_LT_code}}
Let $\{u_1, \ldots, u_n\}$ be $n$ information symbols over $\mathbb{F}_q,  q=2^p$.
The reason why we use the same notation $u_i$, that we used for coded symbols in the previous subsection, is that later we are going to combine a linear $[k,n]$ code and $[N,n]$ LT code to get a raptor code for our blockchain.
The summation $\oplus$ of such symbols is conducted by bitwise XOR operation. Note that since $q=2^p$, the subtraction $\ominus$ coincides with $\oplus$. 
%
With an LT code, one generates 
each coded symbol $v_j$ as follows.  
It first chooses a degree $d$ from a {\em generator degree distribution} $\Omega(d)$. It then randomly chooses $d$ information symbols  $\{ u_{i_1}, \ldots, u_{i_d} \}$
from $\{u_1, \ldots, u_n\}$ with equal probability. These $d$ information symbols are {\em neighbors} of $v_j$. We denote the set of indices of these neighbors by 
\begin{align}
    \mathcal{N}(v_j) = \{i_1, \ldots, i_d \}. 
\end{align}
Lastly, the coded symbol is calculated via the XOR operation
\begin{align}
\label{equ_LT_encode}
    v_j = \sum_{i \in \mathcal{N}(v_j)} u_i =  u_{i_1} \oplus \ldots \oplus u_{i_d}. 
\end{align}
The coded symbol $v_j$ is the {\em edge} of the information symbols $\{ u_{i_1}, \ldots, u_{i_d} \}$. The set of all edges of $u_i$ is denoted by 
\begin{align}
    \mathcal{E}(u_i) = \{j : i \in \mathcal{N}(v_j) \}. 
\end{align}
An information symbol $u_i$ can be calculated using an edge coded symbol and its other neighboring information symbols as
\begin{align}
\label{equ_repair_1}
    u_i = v_j \oplus \sum_{h \in \mathcal{N}(v_j) \backslash \{i\}} u_h, \text{ for any }  j \in \mathcal{E}(u_i).
\end{align}
An example of encoding four information symbols into five LT coded symbols is shown in Fig. \ref{ratelss_intro_LT}. In this example we have 
$\mathcal{N}(v_1) = \{u_1\}$, $\mathcal{N}(v_2) = \{u_1, u_3\}$, $\mathcal{N}(v_3) = \{u_2, u_3, u_4\}$ and so on; and $\mathcal{E}(u_1) = \{v_1\}$, $\mathcal{E}(u_2) = \{v_3, v_4\}$, and so on.  

\begin{figure}[h]
\centering
   \includegraphics[scale = 0.5]{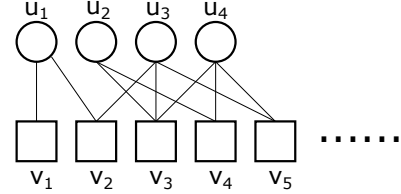}
  \caption{LT encoding: 4 information symbols (circles) are encoded into 5 coded symbols (squares) with degrees $d = 1, 2, 3, 2,  2, $ respectively.  }
  \label{ratelss_intro_LT}
\end{figure}

The receiver can start decoding when it receives $(1+\epsilon)n$ coded symbols. The decoding of the LT code is a `peeling' process. 
%
First, the decoder finds a coded symbol of degree one, that is $v_j$ with $|\mathcal{N}(v_j)| = 1$. This coded symbol in fact is an information symbol with index $\{i\} = \mathcal{N}(v_j)$. 
Having decoded the information symbol $u_i$, the decoder removes $u_i$ from the set of neighbors of all coded symbols from $\mathcal{E}(u_i)$, updates these coded symbols,  and also decreases by $1$ their degrees. 
For example, if $j' \in \mathcal{E}(u_i)$, then $v_{j'}$ and its set of neighbors are updated as follows 
\begin{align}
    v_{j'} \leftarrow v_{j'} \oplus u_i, \quad \mathcal{N}(v_{j'}) \leftarrow \mathcal{N}(v_{j'})\backslash \{i\}.
\end{align}
The decoder continues finding coded symbols of degree one and `peeling' the information symbols until no coded symbol of degree one left or all information symbols are successfully decoded.
%
If there are still undecoded information symbols, the receiver collects more coded symbols and restart the `peeling' process. For example, as shown in Fig. \ref{ratelss_intro_LT_decode}, for the LT code in Fig. \ref{ratelss_intro_LT}, the decoder can successfully decode information symbols $\{u_1, u_3 \}$ after it receives $\{v_1, v_2, v_3, v_4\}$. However, it needs to collect one more coded symbol $v_5$ to successfully decode $\{u_2, u_4\}$.

\begin{figure}[h]
\centering
   \includegraphics[scale = 0.5]{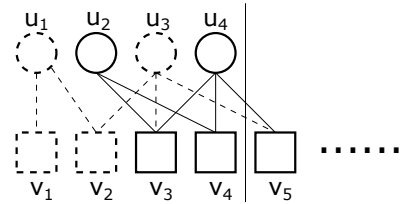}
  \caption{LT decoding: dashed circles indicate decoded information symbols. Dashed squares indicate coded symbols with zero degree. Dashed lines indicate the subtracted degree.   }
  \label{ratelss_intro_LT_decode}
\end{figure}

An issue with the LT code is the error floor \cite{mackay2005fountain} \cite{hussain2011error}. The error floor means that the decoding failure rate does not drop sufficiently fast with more received coded symbols. This happens because the neighbors of coded symbols are uniformly chosen between $\{1, \ldots, n \}$, and thus, there is a non-zero probability that an information symbol is not the neighbor of any received coded symbol, i.e., its edge is an empty set. 
%

\subsubsection{Raptor code \cite{shokrollahi2006raptor}}
In order to address the error floor issue and to reduce the encoding/decoding complexity of the LT code, a raptor code first pre-codes the information symbols using a fixed rate outer code, such as LDPC, Reed-Solomon, or another linear code. The pre-coded symbols are called the {\em intermediate symbols}, which are then encoded into {\em coded} LT symbols.  
The raptor codes have very low decoding failure probabilities, and thus they fit well with the stringent security requirement of blockchains. 
Moreover, the encoding and decoding complexity of raptor codes are $\mathcal{O}(k \log (\frac{1}{\epsilon}))$ \cite{shokrollahi2006raptor}, which is approxiamtely a linear function of $k$.
In what follows, we use raptor codes for our coded blockchains.

\section{Overview of Proposed Rateless Coded Blockchain \label{sec_coding_scheme}}

We first define the dynamic IoT blockchain network as follows. Time is divided into time epochs and $\mathcal{V}_t$ denotes the set of network nodes available in the network in the $t$-th time epoch. 
Each node in the network can send `get address' message recursively to find all reachable nodes in the network to obtain $\mathcal{V}_t$ \cite{bitcoin_data}. Moreover, an address can be considered as valid only if it holds valid UTXOs to prevent Sybil attack. 
We assume that the blockchain grows at a consistent rate of $\beta$ blocks per time epoch. 
We further assume that in each time epoch $t$, $l_t$ nodes leave and $e_t$ nodes join the network, where $l_t$ and $e_t$ are Poisson distributed with means $\lambda_l$ and $\lambda_e$, respectively. 
We next present the encoding and decoding procedures of the proposed rateless coded blockchain.
%
%



\subsection{Encoding Procedure \label{subsec_encoding} }

We divide each block $B^w$ into $s$ symbols over $\mathbb{F}_q$ and form groups as in Fig. \ref{block_groups}. Thus for each group we obtain a matrix in the form of \eqref{equ_original}. For each group we encode $s$ rows of \eqref{equ_original} into $s$ code vectors using the same raptor code for each row. 
%
%
%
The number of blocks in each group varies depending on the network condition (detailed in Section \ref{sec_k}). 
%
We use $\mathcal{G}_m$ to denote the set of blocks in the $m$-th group and $k_m = |\mathcal{G}_m|$. 
In general $k_m$ is large if the number of active nodes in the network is large, and nodes do not actively leave the network. 
We encode the blocks in each group $\mathcal{G}_m$ using a systematic raptor code \cite{shokrollahi2006raptor} with an $[n_m, k_m]$ pre-code with a generator matrix $\mathbf{G}_m$. 
%
We assume that the  {\em code rate} $r=k_m/n_m$ of the pre-code is the same for all groups and time epochs.   
The encoding and decoding procedures are the same for all groups, thus we sometimes may omit the subscript $m$ to simplify the notation. 

Since the same encoding and decoding is used for each row of \eqref{equ_original}, and therefore each row of \eqref{equ_intermedia}, we will write $b_i, i=1,...,k$, and $u_i, i=1,...,n$, meaning any row of  \eqref{equ_original} and \eqref{equ_intermedia} respectively. Similarly for the matrix \eqref{equ_encoded} presented below, we will write $v_i, i=1,...,N$, meaning any row of \eqref{equ_encoded}. 
%
The information symbols $\{b_i \}$ are first encoded into intermediate symbols $\{u_i\}$ using \eqref{equ_original}-\eqref{equ_intermedia}. 
%
%
%
We call each column of \eqref{equ_original}, denoted by $\mathbf{b}_i$, an {\em original block}; and each column of \eqref{equ_intermedia}, denoted by $\mathbf{u}_i$, an {\em intermediate block}. 
%
%
%
%
We assume that $\mathbf{G}_m$ is in the systematic form, and therefore $\{\mathbf{u}_{1}, \ldots, \mathbf{u}_{k}\} = \{\mathbf{b}_{1}, \ldots, \mathbf{b}_{k}\}$. 

The intermediate symbols of each row of \eqref{equ_intermedia} are then encoded into coded symbols using a systematic LT code with a degree distribution $\Omega(d)$. As a result we obtain the matrix 
\begin{align}
\label{equ_encoded}
\begin{bmatrix} 
    v_{1, 1}  & v_{2, 1}  & \ldots & v_{k, 1} & \ldots & v_{n, 1}  & \ldots & v_{N, 1} \\
    v_{1, 2}  & v_{2, 2}  &\ldots & v_{k, 2}  & \ldots & v_{n, 2}  & \ldots & v_{N, 2} \\
    \vdots  &  \vdots & \vdots & \vdots   &  \vdots  &  \vdots  &  \vdots  &  \vdots \\
    v_{1, s}  & v_{2, s}  & \ldots & v_{k, s}  & \ldots & v_{n, s}  & \ldots & v_{N, s}
    \end{bmatrix}.
\end{align}
Here $N$ is the number of network nodes at the moment of encoding. Each column $j$ of \eqref{equ_encoded} is a {\em coded block} denoted by $\mathbf{v}_j$. 
We will use $\mathcal{E}(\mathbf{u}_i)$ to denote the set of edges for symbols $\{u_i\}$ and $\mathcal{N}(\mathbf{v}_i)$ to denote the set of neighbors for symbols  $\{v_i\}$, respectively.
%
We assume systematic encoding and therefore we have {\em systematic coded blocks}  $\{\mathbf{v}_{1}, \ldots, \mathbf{v}_{n}\} = \{\mathbf{u}_{1}, \ldots, \mathbf{u}_{n}\}$ and {\em non-systematic coded blocks} $\{\mathbf{v}_{n+1}, \ldots, \mathbf{v}_{N}\}$. A coding graph from $k$ original symbols $\{b_i\}$ to $n$ intermediate symbols $\{u_i \}$ and then the $N$ coded symbols $\{v_j \}$ is shown in Fig. \ref{encoding_graph}. 
Each node stores only one coded block for each group. Note that, the indices of the coded block that a node stores for different groups may be different. For example, a node may store $\mathbf{v}_a$ for the $m$-th group and $\mathbf{v}_b$ for the $m'$-th group with $a \neq b$. 
However, to make our description simple, we assume that all nodes are enumerated by integers and that node $j$ stores $\mathbf{v}^m_j$ for all groups $m = 1, ...., M$, where $M$ is the total number of groups. 
In the proposed rateless coded blockchain, we use a degree distribution $\Omega(d)$ with $\Omega(1)=0$. An example will be considered in Section \ref{sec_simulation}.

\begin{figure}[t]
\centering
  \includegraphics[scale = 0.25]{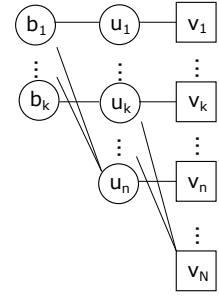}
  \caption{Block coding graph using systematic raptor code.  }
  \label{encoding_graph}
\end{figure}

\subsection{Decoding Procedure \label{sec_decoding} }

When an original block $\mathbf{b}_i$ from group $m$ is needed by a node $j$, node $j$ broadcasts a request to all nodes in the network. Then node $i$ sends $\mathbf{v}_i$  from group $m$ (note that $\mathbf{v}_i = \mathbf{b}_i$) to
node $j$. 
However, the blockchain network is dynamic and node $i$ may leave the network or be disconnected.
When this happens, node $j$ can `repair' the needed block using the procedure that will be presented in Section \ref{sec_block_repair}. 
Briefly, node $j$ `repairs' the needed block using \eqref{equ_repair_1} by obtaining some other coded blocks from active nodes. 
If it fails, node $j$ decodes all blocks in group $m$. To do this, the node $j$ collects $(1+\epsilon)k$ distinct coded blocks from other nodes, which are the columns of  \eqref{equ_encoded}. Since the non-systematic coded blocks $\{\mathbf{v}_{n+1}, \ldots, \mathbf{v}_{N}\}$ have random number of neighbors and the neighbors are randomly chosen, with a very high probability arbitrary $(1+\epsilon)k$ coded blocks are distinct. The decoder first runs the LT decoding as described in Section \ref{sec_LT_code} to decode the intermediate blocks. It then runs a decoder of the pre-code to obtain the original blocks. 
Since each row of \eqref{equ_original} is encoded into the corresponding row of \eqref{equ_encoded} according to the same generator matrix of the pre-code and the same LT encoding procedure, the decoder can successfully decodes all original blocks if it can decode the first row of \eqref{equ_original}.

We would like to note that since we use an LT code, all nodes store distinct coded blocks with very high probability, even if the number of nodes $N$ keeps growing.  
%
This reduces the chance that multiple nodes store the same coded block. 
Moreover, the LT encoding procedure
is much simpler than the design of a $k \times N$ generator matrix need for usual (not rateless) linear codes, especially when $N$ is large and not known in advance. 
%
%
When a new node joins the system, it only needs to collect on average $\mathbb{E}[\Omega(d)]$ systematic coded blocks to encode its own coded block to store for each group, see the details in Section \ref{sec_NMA}. 
This significantly reduces the communication overhead as compared with \cite{kadhe2019sef}.
The code length $k_m$ is calculated based on the network condition. 
If we expect that the network conditions are going to deteriorate (for example we expect that many nodes may leave the network), then we can decode and further re-encode all previously encoded groups with small value $k_m$. 
All nodes in the network must have a consensus on the blocks included into group $m$ and the generator matrix $\mathbf{G}_m$ of the pre-code. 
%
To provide such consensus, we propose a novel {\em enhanced block} in Section \ref{sec_code_block} that can be easily embedded into the current blockchain. 


\begin{figure}[t]
\centering
\subfigure[Rateless coded blockchain \label{Logic_cain}]{\includegraphics[scale = 0.45]{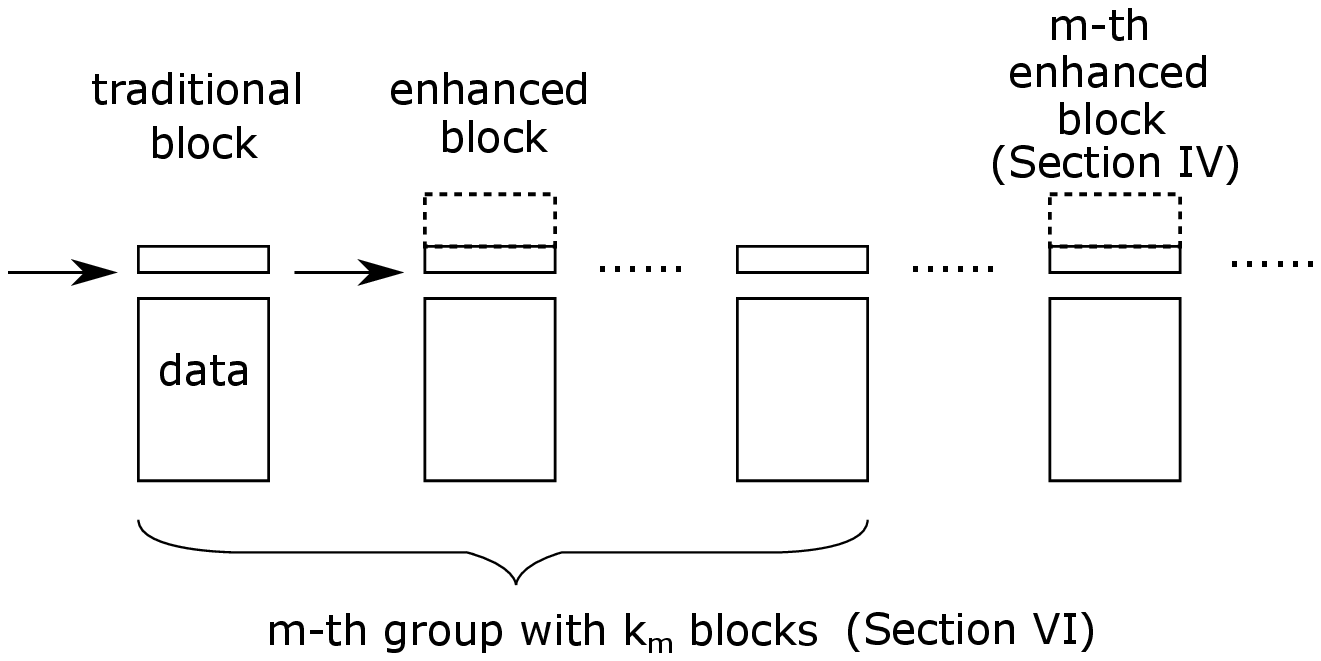}}\hspace*{-0em} \\
\subfigure[New joining node calculates coded block it needs to store. \label{new_join}]{\includegraphics[scale = 0.45]{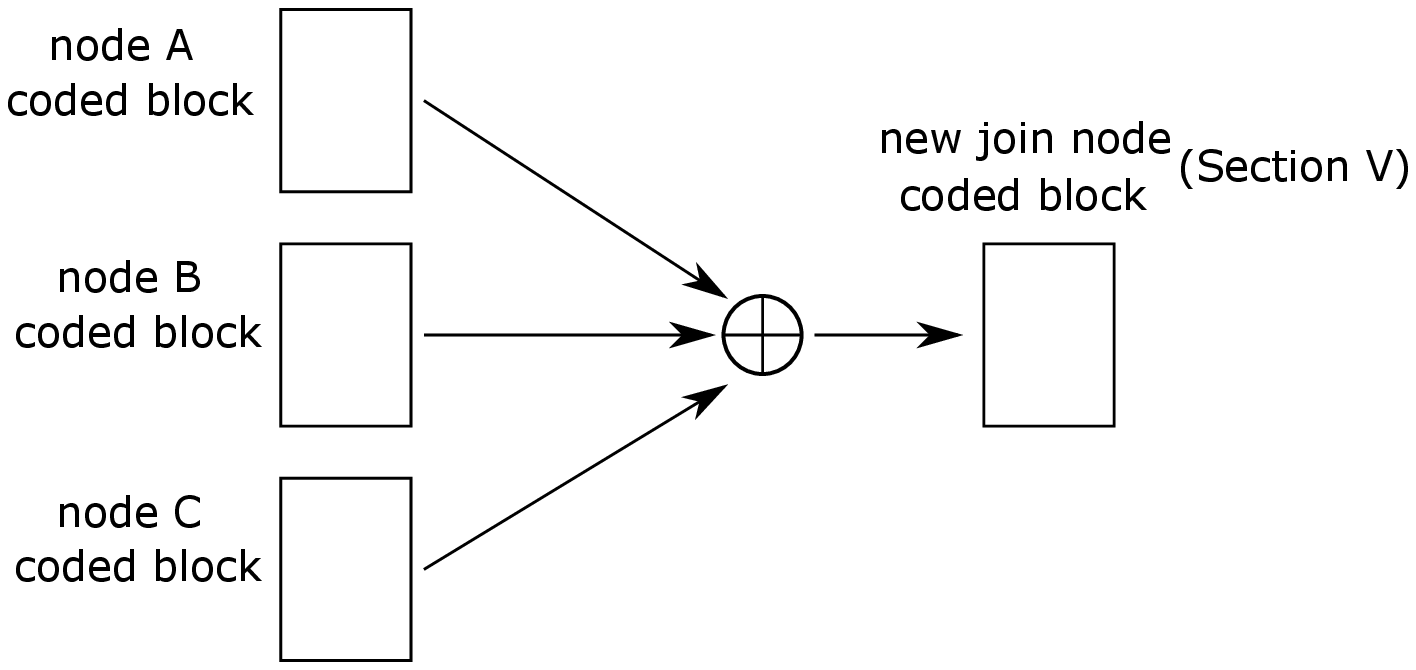}}
\caption{An overview of the proposed rateless coded blockchain.}
\label{blockchains}
\end{figure}

An overview of the proposed coded blockchain is shown in Fig. \ref{Logic_cain} and \ref{new_join}. 
In Fig. \ref{Logic_cain}, data and headers of the blocks that form the $m$-th group are shown, together with the $m$-th enhanced block that stores coding parameters for group $m$. 
The $m$-th enhanced block is always generated later than any blocks in the $m$-th group.
%
Details on how enhanced blocks are organized and how we determine the group size $k_m$ are given in 
Section \ref{sec_code_block} and Section \ref{sec_k}, respectively. 
%
When a new node joins the network in Fig. \ref{new_join}, it only needs to collect a small number of coded blocks from other nodes to calculate the coded block to store. Details of the new node joining procedure will be given in Section \ref{sec_NMA}.

\begin{remark}
The proposed encoding scheme is visibly different from the standard raptor code encoding, for instance presented in \cite{shokrollahi2006raptor}. In the standard scenario one simply  produces a semi-infinite sequence of parity symbols, since systematic  (information) symbols are always available and never erased. In our case, however, systematic symbols may become unavailable if some systematic nodes leave the network. This result in significant changes of encoding procedure, which takes into account possible absence of some systematic symbols (see Section V). This also results in that new (joining) nodes generate code symbols according to a different degree distribution, compared to the degree distribution used at the first instance of encoding a particular codeword (see Lemma 1). Up to the best of our knowledge raptor codes in such regime have not been studied yet. Since our encoding and the effective erasure channel are substantially different from the standard ones, the standard analysis of the code performance, see \cite{shokrollahi2006raptor}, becomes inapplicable. For this reason we proposed a different approach for estimation of the code performance, see Section VI. 
\end{remark}

%
%

\section{Enhanced Block \label{sec_code_block}}

In the existing coded blockchain system, see for example \cite{wu2020distributed}, the value of $k$ is fixed for all groups. 
%
%
However, in our scheme, the group size varies according to the network condition and all nodes must have a consensus on coding parameters used for a particular group. 
Therefore, in this section, we present an enhanced block, which consists of a traditional block and some extra fields in the header to store the coding parameters.
%
%
%
The block structure, mining and verification of the enhanced block is slightly different from traditional blocks. 


\subsection{Enhanced Block Structure \label{sec_CID_structure}}
The data of enhanced blocks is organized in the same way as in usual (not enhanced) blocks. 
However, the enhanced block header contains four additional fields compared with the header of a usual block, see Fig. 6. These new fields are 
%
\begin{itemize}
    \item {\em group sequence number}: this field is the sequence number of the enhanced block, which is independent of the block sequence number. 
    
    \item {\em group indices}: this field contains the block sequence numbers of the blocks to be encoded, i.e., $\mathcal{G}_m = \{w_1, \ldots, w_k \}$. Note that $k_m = |\mathcal{G}_m|$ and $n_m = k_m / r$.
    
    \item {\em generator matrix}: this field contains the generator matrix $\mathbf{G}_m$.
    It should be noted that this field can be omitted if the blockchain system uses a pre-defined generator matrix calculation algorithm. For example, if we use an LDPC code to encode the intermediate blocks, the code can be generated with progressive edge-growth (PEG) algorithm \cite{hu2005regular} 
    and therefore it is enough to store values for $k_m$ and $n_m$ in the {\em group indices} field. 
    
    \item {\em hash values}: this field contains the hash values of the non-systematic intermediate blocks $\{\mathbf{u}_{{k+1}}, \ldots, \mathbf{u}_{n}\}$. %
    When a new node joins the network, it needs systematic and sometimes non-systematic intermediate blocks in order  to generate its own coded block. 
    %
    %
    The integrity of the systematic intermediate blocks, which are in fact the original blocks, can be verified via their merkle tree roots in the corresponding block headers.
    The hash values field ensures that the new node can also validate the integrity of the non-systematic intermediate blocks.
\end{itemize}
%

\begin{figure}[h]
\centering
   \includegraphics[scale = 0.25]{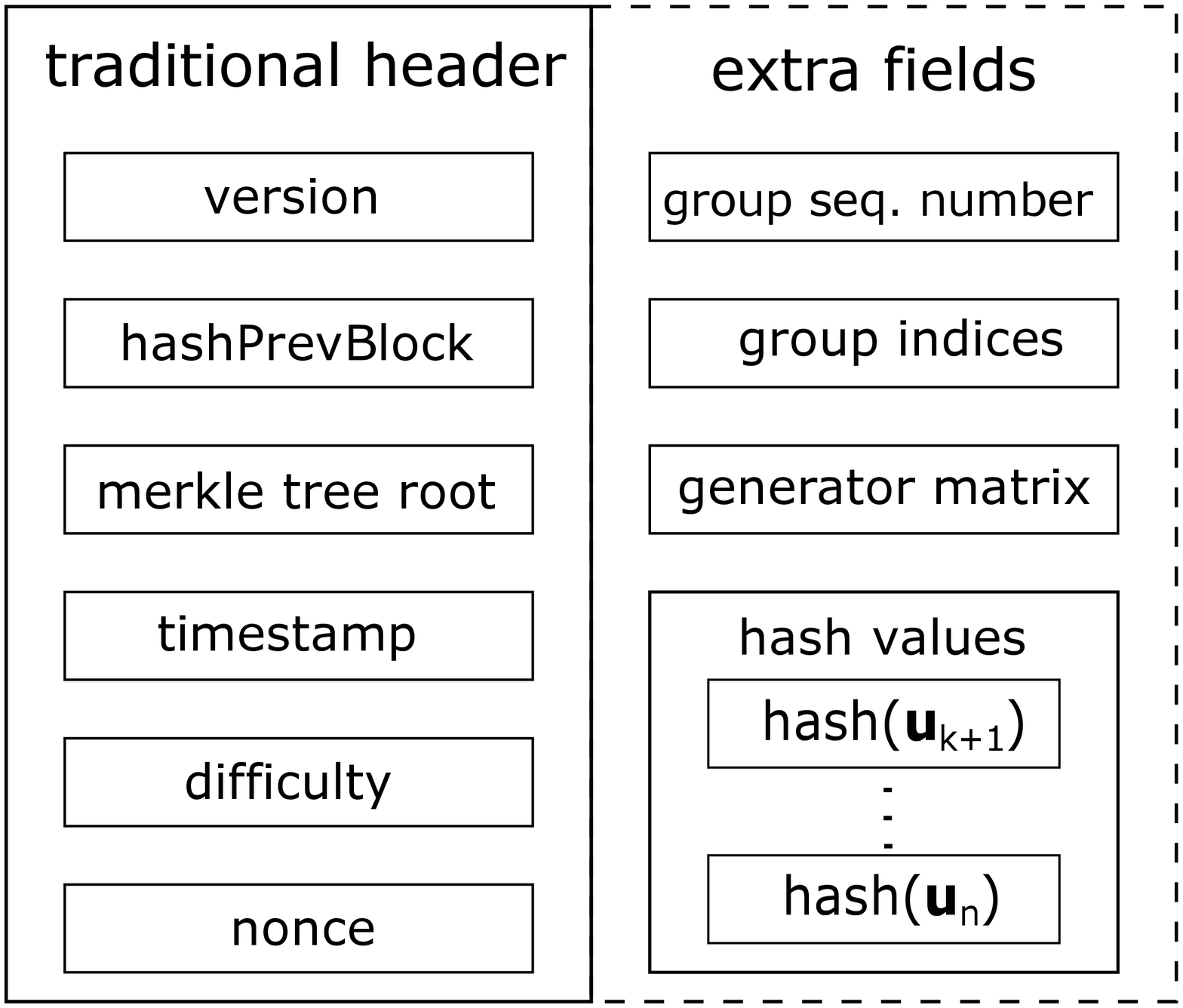}
  \caption{Enhanced block header structure. }
  \label{code_index_block_data_structure}
\end{figure}




\subsection{Mining}

Recall that we employ the PoW protocol in which all nodes competitively `mine' new blocks. 
The `mining' procedure for the traditional blocks is the same as that in the existing blockchain. However, `mining' an enhanced block needs some more efforts. 
Let $M$ be the total number of `mined' enhanced blocks. All nodes in the network use the same algorithm to compute the value $k_{M+1}$, i.e., the size of the next group (see Section \ref{sec_NMA}).  
%
%
Each node $j$ also maintains a block pool $\mathcal{P}_j$, which contains confirmed, but not yet encoded blocks, see Fig. \ref{blockchain_model}.
%
%
A node $j$ starts to `mine' the $(M+1)$-th enhanced block when $|\mathcal{P}_j| \ge k_{M+1}$. This includes the following steps. 

\begin{enumerate}
    \item The `miner' first chooses transactions to include into the enhanced block and calculates the merkle tree root. It then fulfills the fields of version, hash of previous block, timestamp and difficulty. So this part is similar to  `mining' a traditional block. 
    
    \item The `miner' forms the set $\mathcal{G}_{M+1} \subset \mathcal{P}_j$ with size $|\mathcal{G}_{M+1}| = k_{M+1}$. 
    
    \item  The `miner' then forms an $k_{M+1} \times n_{M+1}$ generator matrix $\mathbf{G}_{M+1}$.  
    
    \item Next, the `miner' calculates the intermediate blocks $\{\mathbf{u}_{1}, \ldots, \mathbf{u}_{n_{M+1}}\}$ using \eqref{cal_inter_blocks} and their hash values. 
    
    
    \item Lastly, similar to `mining' a traditional block, the `miner' finds the nonce such that the hash value of the enhanced block header is less than the difficulty level.
\end{enumerate}
After the above procedure, the `miner' produces an enhanced block with the header structure described in Section \ref{sec_CID_structure}. We see that steps 2-4 are not present in the mining of a traditional block. 
However, the computational complexity of these steps is negligible compared with finding the nonce. Hence, the incentive for `mining' an enhance block is the same as that for mining a traditional block.

\begin{figure}[h]
\centering
   \includegraphics[scale = 0.4]{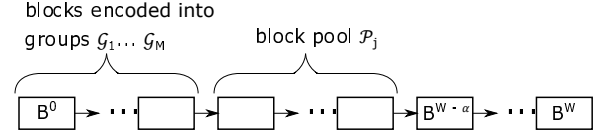}
  \caption{Coded blockchain with block pool.}
  \label{blockchain_model}
\end{figure}

%

\subsection{Verification}
When a node receives a new `mined' enhanced block with sequence number $M+1$, the node verifies this block. If the new enhanced block is valid, the node attaches it to its chain and terminates its own `mining' process of the $(M+1)$-th enhanced block. Otherwise, it continues its own `mining' process. In particular, the verification at node $j'$ should include, but not limited to, the following steps:
\begin{itemize}
    \item The enhanced block is correctly formed including block data and block header structure. This step can include the same verification process for traditional blocks, such as checking nonce and target, size of each field, chain to the previous block, etc. 
    \item The block pool size of the verifying node $j'$ is larger than or equal to $k_{M+1}$, i.e.,
    \begin{align}
        |\mathcal{P}_{j'}| \ge k_{M+1}. \label{cons_pool_size}
    \end{align}
    \item All blocks in the encoding group are within the block pool, i.e.,
    \begin{align}
        \mathcal{G}_{M+1} \subset \mathcal{P}_{j'}. 
    \end{align}
    \item Calculating $\{\mathbf{u}_{k_{M+1}+1}, \ldots, \mathbf{u}_{n_{M+1}}\}$ using \eqref{cal_inter_blocks} and $\mathbf{G}_{M+1}$ and verifying their hash values. 
\end{itemize}
We recall that when mining an enhanced block, a node needs slightly more resources, because it has to find $k_{M+1}$ and $\mathbf{G}_{M+1}$. Though this extra resources are much smaller than resources needed for finding a nonce, still some nodes may continue mining traditional blocks instead of an enhanced block. To prevent this, we
propose to use a punishment mechanism. For example, nodes may follow the rule that if a node has a block pool size satisfying \eqref{cons_pool_size}, then it accepts a new block if and only if this is an enhanced block.


\subsection{Block Pool Update}

After a node $j$ `mined' or validated the new $(M+1)$-th enhanced block, its block pool is updated as
\begin{align}
    \mathcal{P}_j \leftarrow  \mathcal{P}_j \backslash \mathcal{G}_{M+1}.
\end{align}
However, it should be noted that the blocks from $\mathcal{G}_{M+1}$ are not yet ready for encoding since the depth of the $(M+1)$-th enhanced block is less than $\alpha$, i.e., this enhanced block is not confirmed. 
It  could happen that the blockchain splits into several branches, and that different candidate $(M+1)$-th enhanced blocks are `mined' in more than one branch.
In this case, node $j$ will maintain pools $\mathcal{P}^f_j$ for each branch $f$ that contains a candidate $(M+1)$-th enhanced block. As soon as one of these candidate enhanced blocks is confirmed, say in branch $f^*$, the node $j$ assigns $\mathcal{P}_j = \mathcal{P}^{f^*}_j$ and discards the block pools in other branches. In this way the node keeps the pool of available blocks associated with the longest branch.

\subsection{Decentralized Systematic Raptor Code Encoding}

%
After the $m$-th enhanced block is confirmed, all nodes start to encode the blocks in $\mathcal{G}_m$ according to the procedure described in Section \ref{subsec_encoding}, but in a distributive manner. 
We note that, because of network propagation delays and different processing speeds, at a given time moment $t$ different nodes typically have different values of $W$. Let $\hat{B}(m)$ be the block sequence number of the $m$-th enhanced block. Let us assume that any node $j$ has $W$ such that  $W-\alpha = \hat{B}(m)$. Let us denote by
\begin{align}
    \mathcal{C}^m_j = \{ i: \exists j',  \mathcal{N}(\mathbf{v}^m_{j'})=\{i\} \},
\end{align}
the set of intermediate blocks that are stored by some node at a given moment $t$. We describe the way how node $j$ gets $\mathcal{C}^m_j$ later in this subsection.  
%
%
Note that it is possible  that $\mathcal{C}^m_j = \emptyset$. It is also possible that sets $\mathcal{C}^m_j$ are different for different nodes. 
Let us also define
\begin{align}
    \mathcal{S}^m_j = \{ 1, \ldots, n \} \backslash \mathcal{C}^m_j,
\end{align}
the set of intermediate blocks that are not stored by any nodes. 

As soon as $W-\alpha = \hat{B}(m)$, node $j$ starts encoding its own coded block for the $m$-th group as follows. 
\begin{itemize}
    \item Node $j$ first checks $\mathcal{S}^m_j$. If $\mathcal{S}^m_j = \emptyset$, then the node generates $d_j$ from the generator degree distribution $\Omega(d)$, and a set $\mathcal{N}(\mathbf{v}^m_j)$, and computes $\mathbf{v}^m_j$ according to \eqref{equ_LT_encode}. If $\mathcal{S}^m_j \neq \emptyset$, the node takes any $i_j \in \mathcal{S}^m_j$, assigns $\mathcal{N}(\mathbf{v}^m_j) =\{i_j\}$, and broadcasts the message ``I store intermediate block $\mathbf{u}^m_{i_j}$ with a timestamp $t_j$'' across the network.
    
    \item Upon receiving this message, each node, say node $j'$, updates its $\mathcal{C}^m_{j'} \leftarrow \mathcal{C}^m_{j'} \cup \{i\} $ and $\mathcal{S}^m_{j'}$. However, if node $j$ receives $i_{j'} = i_j$ from node $j'$ with timestamp $t_{j'} < t_j$, this means node $j'$ decided to store $\mathbf{u}^m_{i_j}$ earlier than node $j$. In this case node $j$ updates $\mathcal{S}^m_{j}$ and $\mathcal{C}^m_j$, and starts the encoding process from scratch. 
    
\end{itemize}


All network nodes follow the above procedure. As a result, this procedure ensures that each systematic coded blocks $\{\mathbf{v}^m_{1}, \ldots, \mathbf{v}^m_{n} \} = \{\mathbf{u}^m_1, \ldots, \mathbf{u}^m_n \}$ is stored at one and only one node. 
Moreover, this procedure also ensures that the degree distribution of the non-systematic coded blocks $\{\mathbf{v}^m_{n+1}, \ldots, \mathbf{v}^m_{n'} \}$ follow the given degree distribution $\Omega(d)$. 
It is possible that some nodes storing the systematic coded block leave the network. 
Moreover, it is possible that malicious nodes pretend to have such systematic blocks, but cannot provide them to honest nodes.
In the next section, we present a repair procedure to dynamically maintain the systematic property.

\section{Network Maintenance Algorithm \label{sec_NMA} }

%

In this section, we present our network maintenance algorithm (NMA). When a new node $j$ joins the network, it first attempts to encode a new non-systematic coded block by obtaining intermediate blocks from other nodes in the network. If it fails, i.e., one or more intermediate blocks are not stored by any node in the network, node $j$ then tries to repair and store a missing intermediate block. 

\subsection{Encoding a Non-systematic Coded Block}
Let $\mathcal{M}$ be the set of enhanced blocks such that all groups of blocks $\mathcal{G}_m, \forall m \in \mathcal{M}$ are encoded. 
When a new node $j$ joins the network, it first copies the original uncoded blocks from other nodes. It then calculates the coded block $\mathbf{v}^m_j$ to store for each $m \in \mathcal{M}$. In particular, for each group $m \in \mathcal{M}$, the node $j$ generates a degree $d^m_j$ from $\Omega(d)$ and forms $\mathcal{N}(\mathbf{v}^m_j)$. 
It then broadcasts across the network a request for blocks with indices from $\mathcal{N}(\mathbf{v}^m_j)$. Upon receiving this request, a node $j'$ will send to node $j$ an index $i$ if $\mathcal{N}(\mathbf{v}^m_{j'}) = \{i\}$ 
for any $i \in \mathcal{N}(\mathbf{v}^m_j)$, to inform node $j$ 
%
that the intermediate block $\mathbf{u}^m_i$.
Node $j$ then assigns $\phi(\mathbf{u}^m_i) = 1$ to indicate the availability of $\mathbf{u}^m_i$, otherwise $\phi(\mathbf{u}^m_i) = 0$. 
Hence, for each group $m \in \mathcal{M}$, node $j$ is able to construct a set $\mathcal{R}^m_j$ of missing intermediate blocks
\begin{align}
    \mathcal{R}^m_j = \{i : \phi(\mathbf{u}^m_i) = 0, i \in  \mathcal{N}(\mathbf{v}^m_j) \}
\end{align}
%
%
If $\mathcal{R}^m_j = \emptyset$, then node $j$ can obtain all blocks from $\mathcal{N}(\mathbf{v}^m_j)$ from other nodes and then it 
%
calculates and stores the coded block $\mathbf{v}^m_j$ according to \eqref{equ_LT_encode}. 
Otherwise, if $\mathcal{R}^m_j \neq \emptyset$ due to nodes storing the systematic coded blocks leaving the network or having bad network connection,
node $j$ will try to repair a block in $\mathcal{R}^m_j$. It starts by trying to repair one block from $\mathcal{R}^m_j$. 
During the process of repairing, the set $\mathcal{R}^m_j$ can be expanded with other missing intermediate blocks, as it is described in the next Section.


\subsection{Intermediate Block Repair \label{sec_block_repair} }

Node $j$ chooses random $i \in \mathcal{R}^m_j$ and tries to repair $\mathbf{u}^m_i$ using an edge coded block $\mathbf{v}^m_{j'}$, where $j' \in \mathcal{E}(\mathbf{u}^m_i)$, and $\mathbf{v}^m_{j'}$'s other neighboring intermediate blocks $\mathcal{N}(\mathbf{v}^m_{j'}) \backslash \{i\}$ similar to \eqref{equ_repair_1}, i.e.,
\begin{align}
\label{equ_repair_2}
    \mathbf{u}^m_i = \mathbf{v}^m_{j'} \oplus \sum_{h \in \mathcal{N}(\mathbf{v}^m_{j'}) \backslash \{i\}} \mathbf{u}^m_h, \text{ for any }  j' \in \mathcal{E}(\mathbf{u}^m_i).
\end{align}
%
%
%
However, the set of edges $\mathcal{E}(\mathbf{u}^m_i)$ and the set of neighbors $\mathcal{N}(\mathbf{v}^m_{j'})$ for any $j' \in \mathcal{E}(\mathbf{u}^m_i)$  are not known by node $j$.
For this reason, node $j$ tries to form a subset $\hat{\mathcal{E}}(\mathbf{u}^m_i) \subseteq  \mathcal{E}(\mathbf{u}^m_i)$. Initially, it sets up $\hat{\mathcal{E}}(\mathbf{u}^m_i)  = \emptyset$ and broadcasts $i$ across the network. 
%
When another node $x$ receives this message, it will reply to node $j$ with $\mathcal{N}(\mathbf{v}^m_{x})$ if $i \in \mathcal{N}(\mathbf{v}^m_{x})$. 
Then node $j$ adds $x$ into the set $\hat{\mathcal{E}}(\mathbf{u}^m_i)$. 
%
%
Note that, if a node $j' \in \mathcal{E}(\mathbf{u}^m_i)$ has bad network connection and not replies to $j$, then node $j'$ is considered as missing and we have $j' \in  \mathcal{E}(\mathbf{u}^m_i) \backslash \hat{\mathcal{E}}(\mathbf{u}^m_i)$.
After obtaining the set $\hat{\mathcal{E}}(\mathbf{u}^m_i)$, node $j$ checks the availabilities of the intermediate blocks $\mathcal{N}(\mathbf{v}^m_{j'}) \backslash \{i\}$ for each $j' \in \hat{\mathcal{E}}(\mathbf{u}^m_i)$.
%
If there is  $j' \in \hat{\mathcal{E}}(\mathbf{u}^m_i)$ such that all blocks from $\mathcal{N}(\mathbf{v}^m_{j'}) \backslash \{i\}$ are available,
%
then $\mathbf{u}^m_i$ can be repaired according to \eqref{equ_repair_2}. Otherwise, node $j$ adds the missing intermediate blocks into $\mathcal{R}^m_j$, i.e., 
\begin{align}
    \mathcal{R}^m_j \leftarrow \mathcal{R}^m_j \cup \{i' : \phi(\mathbf{u}^m_{i'}) = 0, i' \in \mathcal{N}(\mathbf{v}^m_{j'}) \backslash \{i\}, \nonumber \\
    j' \in \hat{\mathcal{E}}(\mathbf{u}^m_i)\},
\end{align}
and also adds the index $i$  of $\mathbf{u}^m_i$ to the set $\mathcal{Q}^m_j$ of not repairable symbols.
Next node $j$ tries to repair a block from $\mathcal{R}^m_j \backslash \mathcal{Q}^m_j$. 


\begin{figure}[t]
\centering
   \includegraphics[scale = 0.5]{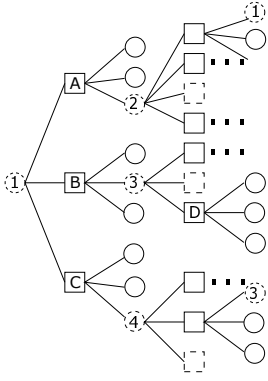}
  \caption{An example of expanded tanner graph with $d=4$ rooted at $\mathbf{u}^m_1$. Solid/dash circles denote available/missing information symbols, solid/dash squares denote available/missing parity check symbols. }
  \label{tanner_graph}
\end{figure}

Figure \ref{tanner_graph} illustrates the intermediate block repair process through an example, where we assume that node $j$ generates $\mathcal{N}(\mathbf{v}^m_{j})$ and notices that $\{1\} \in \mathcal{N}(\mathbf{v}^m_{j})$, but $\mathbf{u}^m_1$ is not available, while other blocks from $\mathcal{N}(\mathbf{v}^m_{j})$ are available.  Node $j$ then forms $\mathcal{R}^m_j = \{1\}$. 
Next it contacts other nodes and constructs the set $\hat{\mathcal{E}}(\mathbf{u}^m_1) = \{A, B, C \}$. 
%
However, the intermediate blocks with indices $\{2\} \in \mathcal{N}(\mathbf{v}^m_{A})$, $\{3\} \in \mathcal{N}(\mathbf{v}^m_{B})$ and $\{4\} \in \mathcal{N}(\mathbf{v}^m_{C})$ are not available. Hence, node $j$ updates $\mathcal{R}^m_j \leftarrow \{1\} \cup \{2, 3, 4 \}$ and constructs the set $\mathcal{Q}^m_j = \{1\}$. Next, node $j$ attempts to repair any of the intermediate block $\mathbf{u}^m_2$, $\mathbf{u}^m_3$ and $\mathbf{u}^m_4$. In this example, it successfully repairs $\mathbf{u}^m_3$ and stores the coded block $\mathbf{v}^m_j = \mathbf{u}^m_3$. Hence, the subsequent joining nodes can obtain $\mathbf{u}^m_3$ from node $j$ to repair $\mathbf{u}^m_4$. 

If  $\mathcal{Q}^m_j = \mathcal{R}^m_j$, this means that node $j$ is not able to repair an intermediate block in the current network.  Then node $j$ collects coded blocks from $(1+\epsilon)k_m$ different nodes and decodes the original blocks using the procedure described in Section \ref{sec_coding_scheme}.
After that, it stores one intermediate block $\mathbf{v}^m_j = \mathbf{u}^m_i$ for any $i \in \mathcal{R}^m_j$.

It should be noted that in the above proposed NMA algorithm, new nodes generate  degrees $d$ with distribution $\Omega^*(d) \neq \Omega(d)$. The distribution $\Omega^*(d)$ depends on the history of the network changes from the moment of encoding group $m$ and up to the moment of a new node joining.
%
%
%
Recall that $n_m$ is the number of intermediate blocks at the moment of encoding group $m$ (see Section \ref{subsec_encoding}).
Let $n^*_m$ be the number of nodes storing unique intermediate block for the $m$-th group at the moment when a new node $j$ joins the network. Then $\Omega^*(d)$ is given by the following lemma.

\begin{Lemma}
\label{lemma1}
    The actual distribution $\Omega^*(d)$ resulted from the NMA algorithm is 
    \begin{align}
    \label{equ_dist_case}
        \Omega^*(d) = 
        \begin{cases}
        \Omega(d)  g(n^*_m, d), \quad & \text{for } d \ge 2, \\
        \sum_{d'=2}^{d'=n_m} \Omega(d') (1 -  g(n^*_m, d')), \quad & \text{for } d=1,
        \end{cases}
    \end{align}
    where
    \begin{align}
        g(n^*_m, d) =  \binom{n^*_m}{d} / \binom{n_m}{d}. \label{equ_18}
    \end{align}
\end{Lemma}
\begin{proof}

A new node $j$ generates degree $d$ with probability  $\Omega(d)$ and chooses $d$ nodes among $n_m$ nodes according to the uniform distribution. The probability that these $d$ nodes are among the available $n^*_m$ nodes is 
\begin{align}
        \Pr(d \text{ nodes are available}) = \binom{n^*_m}{d} / \binom{n_m}{d}.
\end{align}
Thus, we get the first line of \eqref{equ_dist_case}. 

The probability that at least one of the chosen $d$ nodes is not available is 
\begin{align}
        & \Pr(\text{some of $d$ nodes are not available}) \nonumber \\\
        & = 1 -  \binom{n^*_m}{d} / \binom{n_m}{d}.
\end{align}
%
In this case, node $j$ starts repairing/decoding and further stores an intermediate block, which means that it stores a block of degree 1.  Thus, we obtain the second line of \eqref{equ_dist_case}.

\end{proof}

  
%

We now examine the computational complexity of the encoding, repairing and decoding procedures. If the new joining node successfully obtains the needed $d \sim \Omega(d)$ intermediate blocks, the expected encoding complexity is $\mathcal{O}(\mathbb{E}[\Omega(d)])$. 
If the new node needs to repair a missing intermediate block by using \eqref{equ_repair_2} with a coded block $\mathbf{v}^m_j$, the computational complexity is $\mathcal{O}(|\mathcal{N}(\mathbf{v}^m_j ) |)$. The number of neighbors $|\mathcal{N}(\mathbf{v}^m_j ) |$, which is the degree of $\mathbf{v}^m_j$, follows the actual distribution $\Omega^*(d)$. From Lemma~\ref{lemma1}, we see that 
\begin{align}
    \Omega^*(d) \le \Omega(d), \quad \forall d \ge 2. 
\end{align}
Hence, the expected computational complexity for repairing is $\mathcal{O}(\mathbb{E}[\Omega^*(d)]) \le \mathcal{O}(\mathbb{E}[\Omega(d)])$. 
When the new node needs to decode all blocks in the $m$-th group, the computational complexity is $ \mathcal{O}(k_m)$ \cite{shokrollahi2006raptor}. 
Note that the decoding complexity is typically significantly larger than the complexity of encoding and repairing.
For example in our simulations, encoding and repairing typically involve dozens of blocks, while decoding involves thousands of blocks depending on the group size.  
Hence, it is computationally expensive if all new joining nodes have to decode all groups, as is the case in \cite{kadhe2019sef}. On the other hand, by using our approach, the probability that a new joining node has to decode one group, rather than simply encode or repair, is less than $6.3 \times 10^{-5}$, based on simulation results in the Bitcoin blockchain network (see Section \ref{sec_sim_bitcon}).

\section{Determining Group Size $k_m$ \label{sec_k}}

As will be explained later, our goal is to choose $k_m$ as large as possible, but so that after $\gamma$ time epochs the original blocks from $\mathcal{G}_m$ can be restored with a decoding failure probability less than a target failure probability $\zeta$. 
A reliable system typically requires the failure probability be less than $10^{-9}$ \cite{li2010realistic}.
%
%
%
%
%
%
%
In our proposed rateless coded blockchain with a set of $\mathcal{M}$ encoded groups, each node stores only $|\mathcal{M}|$ coded blocks and the blocks that have not been encoded yet. 
%
The extra storage for the enhanced block headers, i.e., the hash values of the non-systematic intermediate blocks, is negligible\footnote{ For example, when a group has $k=10000$ blocks and $r = 0.8$, there are 2500 hash values to store in the enhanced block header and needs 0.076 MB, which is negligible as compared with more than 9.76 GB space saving when all blocks are 1 Mb. }.
Hence, we define the {\em storage reduction coefficient} as follows.
%
%
%
\begin{Definition}
The storage reduction coefficient $R_s$ is defined as the ratio of the space needed by an individual node in the proposed rateless coded blockchain over such space in a replicated blockchain, i.e., 
    \begin{align}
    {R}_s = \frac{W - \sum_{m \in \mathcal{M}} k_m + |\mathcal{M}|}{W}. \label{equ_storage}
    \end{align}
\end{Definition}
%
The total number of encoded blocks is $\sum_{m \in \mathcal{M}} k_m$.
%
Hence, the larger are $k_m$'s, the smaller is $|\mathcal{M}|$, and therefore the smaller is the storage reduction coefficient. Thus we would like $k_m$'s to be large. 
%
On the other hand, if $k_m$ is too large and some nodes leave the network, we may not be able to recover (decode) some groups. 
%
To find a good trade-off, we define the decoding failure probability as follows.
%
%
%
%
%

%

\begin{Definition}
    Given the degree distribution $\Omega(d)$, the rate of nodes leaving and joining the network $\lambda_l$ and $\lambda_e$, respectively, and initial number of nodes $N$, the decoding failure probability $f(k, N; \Omega(d), \lambda_l, \lambda_e, \gamma)$ is the probability that a node fails to decode a group of blocks after $\gamma$ time epochs since the group encoding.
\end{Definition}

\begin{figure}[t]
\centering
   \includegraphics[scale = 0.5]{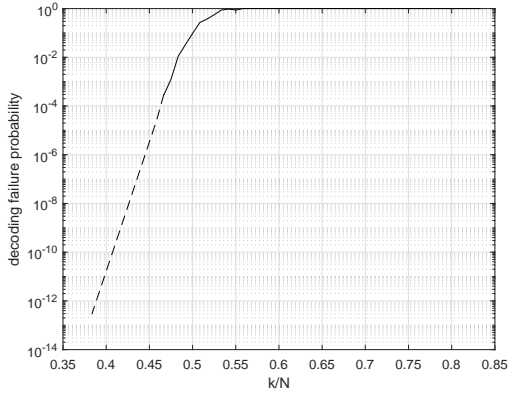}
  \caption{Decoding failure probability of $f(k, 3000; \Omega(d), 12, 4,  100)$. Solid line indicates the failure probability from simulation and dashed line indicates the predicted failure probability. }
  \label{fitting_12_4}
\end{figure}

Let $ p(\mathbf{G}, \Omega^*(d), \epsilon)$ be the probability that the raptor code fails to decode at least one information symbol 
with pre-code matrix $\mathbf{G}$, coded symbol degree distribution $\Omega^*(d)$ and overhead $\epsilon$.
Let us assume that $ p(\mathbf{G}, \Omega^*(d), \epsilon)$ is known to us, for example it can be derived using an analysis similar to that in \cite{shokrollahi2006raptor}, or it can be obtained via simulations. Then we have the following result. 

\begin{Lemma}
\label{lemma_2}
    The decoding failure probability is given by \eqref{equ_f}. 
\end{Lemma}

    
\begin{proof}
    It is not difficult to see that 
    \begin{align}
    \label{equ_25}
        & f(k, N; \Omega(d), \lambda_l, \lambda_e, \gamma) \nonumber \\
        & = \int_{0}^{\infty} p(\mathbf{G}, \Omega^*(d), \epsilon) \Pr(N + \Delta_{\gamma} = (1+\epsilon) k ) \, d \epsilon,
    \end{align}
    where $\Delta_{\gamma}$ is the sum of two Possion processes that follows a Skellam distribution  \cite{karlis2006bayesian}, i.e., 
    \begin{align}
        & \Pr(\Delta_{\gamma} = s) \nonumber \\
        & = e^{-\gamma (\lambda_l + \lambda_e)} \binom{\lambda_l}{\lambda_e}^{\frac{s}{2}} \sum_{X=0}^{\infty}\frac{1}{X!(X + s )!} (\gamma \sqrt{\lambda_l\lambda_e} )^{2 X + s} \label{equ_skellam}
    \end{align}
    Hence, by substituting \eqref{equ_skellam} into \eqref{equ_25}, we obtain equation \eqref{equ_f}. 
    
    \begin{align}
    \label{equ_f}
         &f(k, N; \Omega(d), \lambda_l, \lambda_e, \gamma) \nonumber \\ 
         &= \int_{0}^{\infty} p(\mathbf{G}, \Omega^*(d), \epsilon) e^{-\gamma (\lambda_l + \lambda_e)} \binom{\lambda_l}{\lambda_e}^{\frac{(1 + \epsilon)k - N}{2}} \times \nonumber \\
         & \sum_{X=0}^{\infty}\frac{1}{X!(X + (1 + \epsilon)k - N )!} (\gamma \sqrt{\lambda_l\lambda_e} )^{2 X + (1 + \epsilon)k - N} \, d \epsilon.
    \end{align}

\end{proof}

As we explain below, for finding an appropriate value of $k_m$ we use the function $f(k, N; \Omega(d), \lambda_l, \lambda_e, \gamma)$, which can be obtained in several possible ways. One way is to use analysis similar to the one presented in \cite{shokrollahi2006raptor} for first obtaining values $ p(\mathbf{G}, \Omega^*(d), \epsilon)$  and next using Lemma \ref{lemma_2}. Another way is to obtain $f(k, N; \Omega(d), \lambda_l, \lambda_e, \gamma)$ via simulations, which requires some efforts since obtaining small values of $f(k, N; \Omega(d), \lambda_l, \lambda_e, \gamma)$ may take  too long simulation time.  Several methods for resolving this problem have been proposed by the coding theory community, see for example  \cite{richardson2003error} and references therein. The main goal of our work is not coding theory, however, but design of new blockchain. For this reason, in this work, we simply  
assume that $f(k, N; \Omega(d), \lambda_l, \lambda_e, \gamma)$ is available to us. For our numerical examples considered in the next Section, we generate $f(k, N; \Omega(d), \lambda_l, \lambda_e, \gamma)$ by simulating it up to $10^{-4}$ and extrapolating it further for smaller values. 
For example, when $\lambda_l=12$, $\lambda_e = 4$, $N=3000$, $\gamma=100$, and $\Omega(d)$ follows the generator degree distribution in Section \ref{sec_simulation}, the decoding failure probability is shown in Fig. \ref{fitting_12_4}.

Recall that our goal is to determine $k_m$ at time epoch $t_m$ when the number of nodes in the network is $|\mathcal{V}_{t_{m}}|$, so that the decoding failure probability after $\gamma$ time epoch when network changes to $|\mathcal{V}_{t_m + \gamma}|$ as a result of nodes leaving and joining the network, is no higher than $\zeta$.
%
%
It should be noted that $f(k, N; \Omega(d), \lambda_l, \lambda_e, \gamma)$ becomes a function of $\frac{k}{N}$ when $N \rightarrow \infty$. 
For finite $N$, the probability $f(k, N; \Omega(d), \lambda_l, \lambda_e, \gamma)$ increases as $N$ decreases with a fixed ratio of $\frac{k}{N}$, see simulation results in Fig. \ref{failure_rate}. 
%
%
%
%
%
%
%
%
It should be also noted that the blocks in group $\mathcal{G}_m$ will be encoded after the $m$-th enhanced block is confirmed.  
Recall that $\beta$ is the number of blocks generated in each time interval, which is a constant on average, and $\alpha$ is the required depth for a block to be confirmed.
Therefore, we calculate $k_m$ by table lookup according to the following: 
%
\begin{align}
\label{equ_27}
f(k_m, | {\cal V}_{t_m} |; \Omega(d), \lambda_l, \lambda_e, \gamma + \frac{\alpha}{\beta}) = \zeta
\end{align}
where 
$\frac{\alpha}{\beta}$ is the number of time epochs needed for the $m$-th enhanced block to be confirmed, and $\zeta$ is the target failure probability.

If at a moment $t' \ge t_m + \gamma$ we have that $|\mathcal{V}_{t'}| < |\mathcal{V}_{t_{m}}|$, then we cannot guarantee  that at time $t' + \gamma$ we will have decoding failure probability smaller than $\zeta$.
In this case each node $j$ puts the blocks in $\mathcal{G}_m$ back to the pool $\mathcal{P}_j$, and proceeds further in the standard way with 
one exception - if node $j$ `mines' the next enhanced block, it assigns to it the sequence number $m$, instead of $M+1$.   
%
%
This new $m$-th enhanced block overwrites the old one and its `mining' and verification processes are the same as those for a new `mined' enhanced block. Thus, if more than one enhanced blocks have the same sequence number, only the latest one is valid.
If malicious nodes deliberately `re-mine' an enhanced block with sequence number $m^*$ that not need to be `re-mined', all honest nodes will not accept this enhanced block since the blocks $\mathcal{G}_{m^*}$ are not in their block pools.
The new joining node runs the NMA algorithm in the backward direction - from the latest enhanced block in the longest blockchain branch toward to the genesis block. If the new node bumps to an enhanced block with a sequence number, say $m$, that it already met earlier in this backward process, then the new node simply ignores this enhanced block.

It should be noted that our scheme increases communication overhead when an enhanced block `miner' needs to estimate the the number of nodes in the network. In order to obtain the number of nodes, we can employ the `ping’ and `pong’ mechanism  used in the existing blockchains for checking whether a node is alive \cite{pingpong}. We can only count the nodes that recently broadcast honest messages, such as valid blocks or transactions, to avoid fake connections. Each node can broadcast an alive message every 24 hour, which should be enough since in the modern blockchains the number of nodes typically varies quite slowly. For instance, in the bitcoin network, the average number of joining and leaving nodes per day are 42.18 and 43.16 respectively, see Section VII.C. Hence the percentagewise the increased communication overhead for estimating network size is very smaller compared to the typical network load\footnote{A typical `ping' and `pong' message have 32 byte, which increases network communication overhead $32 \times N^2 / 24$ bytes per hour. A Bitcoin block has $10^6$ bytes and is generated every 10 minutes, which results in a typical network load of $10^6 \times N \times 6$ bytes per hour. Hence the increased communication overhead is 0.22\% when $N=10,000$.}.

\section{Simulation Results \label{sec_simulation} }

Our simulations are conducted using Matlab 2020b on a laptop with an Intel Core i7 eight-core CPU @ 2.2 GHz.
In our simulations, we implemented the proposed rateless coded blockchain using the Reed-Solomon code as the pre-code. The generator matrix of the Reed-Soloman code is from Matlab 2020b Communication Toolbox. The generator degree distribution $\Omega(d)$ is as follows.
Let $S = c \ln (k/\delta) \sqrt{k}$, where $c$ and $\delta$ are constants, and define
\begin{align}
    &\tau(d) = 
    \begin{cases}
    \frac{S}{ik}, &  \text{for } d = 1, \ldots, \frac{k}{S} - 1, \\
    \frac{S \ln (S /\delta)}{k}, &  \text{for } d = \frac{k}{S},  \\
    0, &  \text{for } d = \frac{k}{S} + 1, \ldots, k,  
    \end{cases}\\
    &\rho(d) = 
    \begin{cases}
    \frac{1}{k}, &  \text{for } d = 1, \\
    \frac{1}{d(d-1)}, &  \text{for } d = 2, \ldots, k. \\
    \end{cases}
\end{align}
In \cite{luby2002lt} the following degree distribution is proposed for obtaining LT codes
\begin{align}
\label{equ_mu}
    \mu(d) = \frac{\tau(d) + \rho(d)}{\sum_{j=1}^k \Big(\tau(j) + \rho(j) \Big)}.
\end{align}
In our simulations, we set $c = 0.1$ and $\delta = 0.5$, which are known to lead to good and stable performance \cite{mackay2005fountain}.
As discussed in Section \ref{subsec_encoding}, a new node should use a degree distribution with $\Omega(1)=0$. For this reason we slightly modify $\mu(d)$ in \eqref{equ_mu} and set 
\begin{align}
    \Omega(d) = 
    \begin{cases}
    0, &  \text{for } d = 1, \\
    \mu(d) + \frac{\mu(1)}{k-1}, &  \text{for } d = 2, \ldots, k. \\
    \end{cases}
\end{align}

We first compare the LT and raptor codes in the dynamic network, and show the decoding failure probabilities $f(k, N; \Omega(d), \lambda_l, \lambda_e, \gamma)$ for several network scenarios.
Next, we evaluate the proposed rateless coded blockchain in a highly dynamic blockchain network where the number of nodes leaving is much higher than that for joining. 
Lastly, we show the performance of the proposed scheme in the real-world Bitcoin network.

\subsection{Code Performances}
We first study the minimum number of nodes needed for successful decoding  by using the LT and raptor codes respectively. 
%
%
We set a fixed length blockchain with only one group with $k=1000$, the network initially has $N=3000$ nodes, and $\lambda_l = 50$ and $\lambda_e = 40$ nodes per time epoch. 
We set the code rate of the pre-code as $r = 0.8$. 
The network is dynamic using the NMA algorithm. 
%
We generated 1000 blockchains and ran them until a new node decoding fails. We record the numbers of nodes needed in each of these networks for the latest successful decoding before the failure occur, and show their densities in Figure \ref{density_LT_RP}.  
We see that the raptor code significantly outperforms the LT code since it requires smaller number of nodes for successful decoding. 
We also see that sometimes the LT code needs more than
2400 nodes for decoding, 
which means that sometimes it needs very large, $\epsilon = 1.4$, decoding overhead. 
This is due to the error floor issue present in LT codes \cite{mackay2005fountain} \cite{hussain2011error}. 

\begin{figure*}[htbp]
  \centering
  \mbox{
    \subfigure[{ }\label{density_LT_RP}]
    {\includegraphics[scale = 0.41]{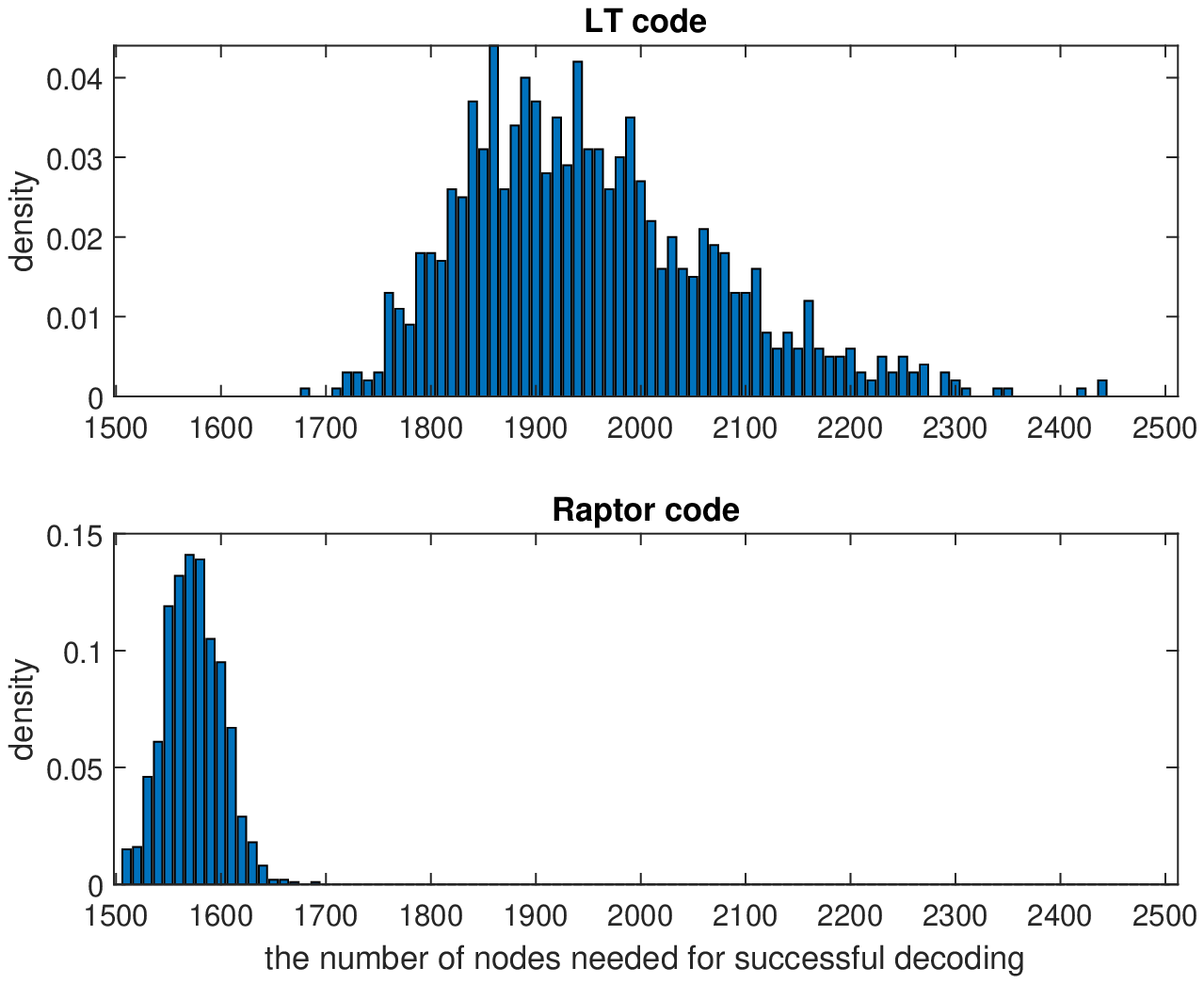}}\hspace*{-0.3em}
     \subfigure[{ }\label{failure_rate}]{\includegraphics[scale = 0.41]{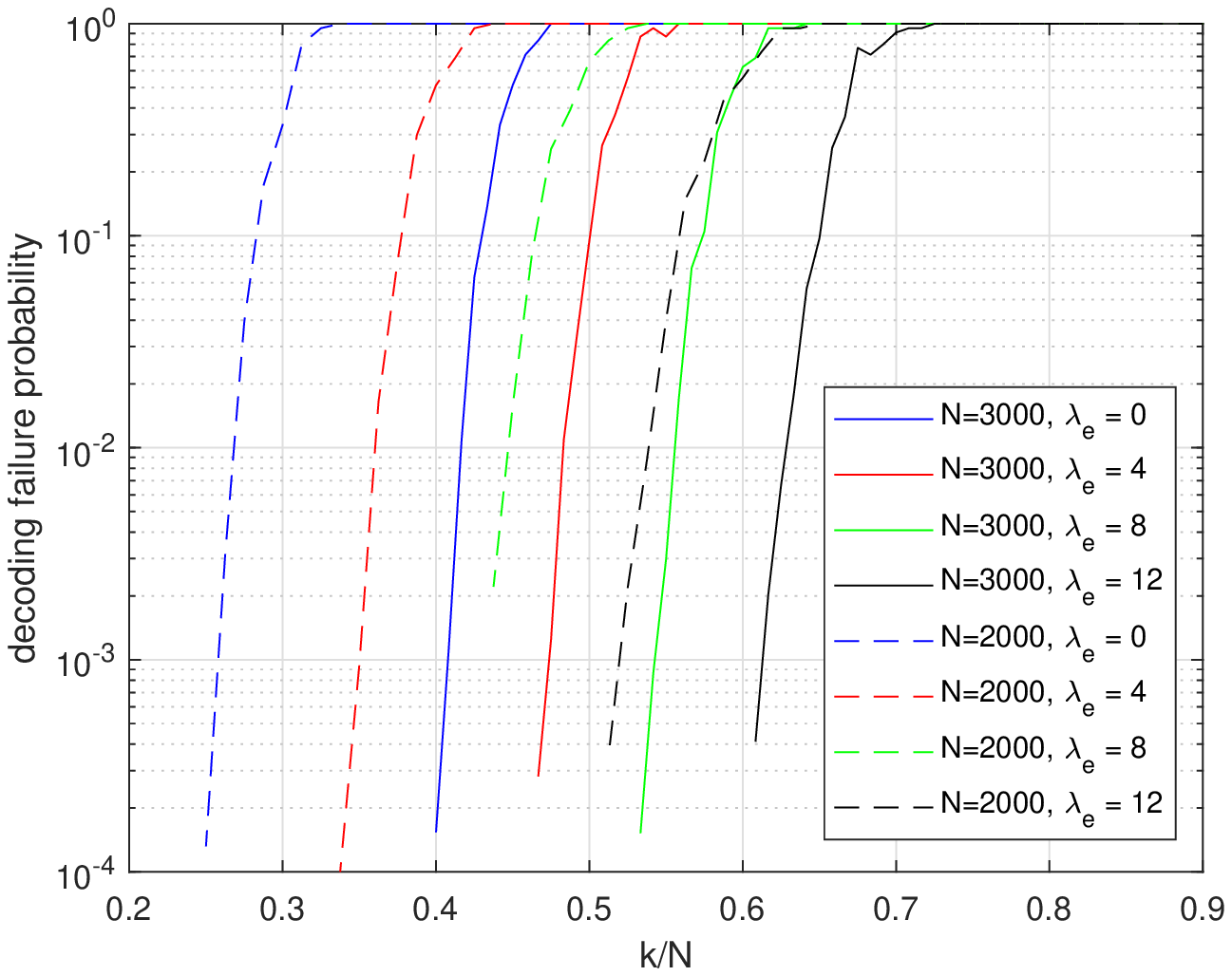}}\hspace*{-0.3em}
     \subfigure[{ }\label{sim_1}]{\includegraphics[scale = 0.41]{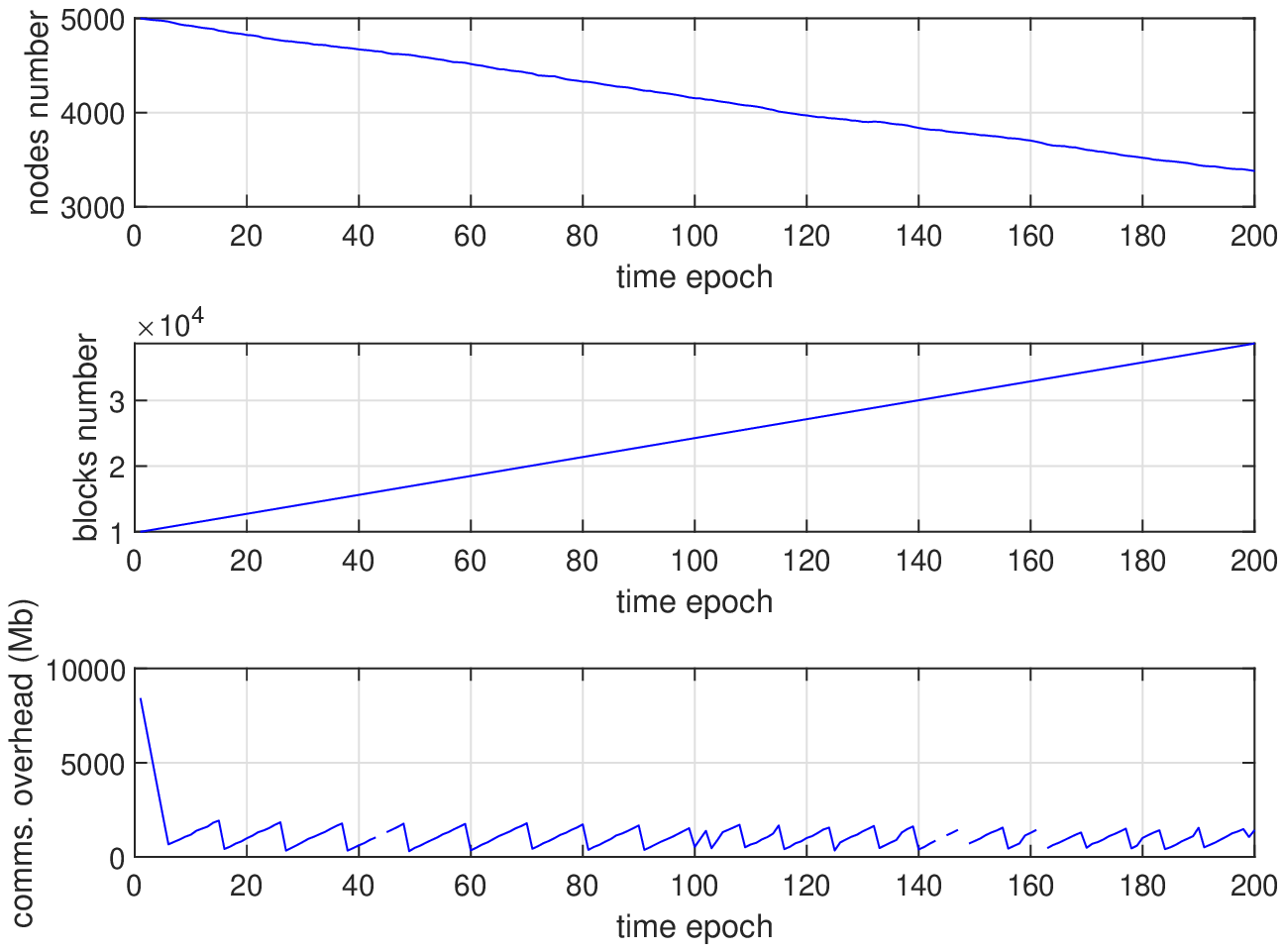}}
   }
  \caption{(a) Performance comparison of LT and Raptor codes in dynamic blockchain networks; (b) Decoding failure probability when $\gamma=100$, $\lambda_{l}=12$, $\lambda_{e}=\{0, 4, 8, 12 \}$; (c) Number of nodes, number of blocks and communication overhead (Mb) of the new joining node with time epoch growth. }
  \label{ThreeFigures2}
\end{figure*}
%


We next show the decoding failure probability $f(k, N; \Omega(d), \lambda_l, \lambda_e, \gamma)$  in Figure \ref{failure_rate}. We fix $\lambda_l = 12$, $\gamma=100$, and take $\lambda_e = \{0, 4, 8, 12 \}$ and $N=\{2000, 3000\}$. 
We keep the code rate of the pre-code as $r=0.8$. 
Note that when $\lambda_e = 0$, decoding blocks in the dynamic network is equivalent to decoding through a binary erasure channel (BEC) since nodes (symbols) leave the network, i.e., are erased, with uniform probability, and no repairing occurs.
From the result, we can see that BEC is the worst scenario in the sense that it leads to the minimum value of $k$ for a given decoding failure probability. 
%
%
This opens the possibility of obtaining $f(.)$ by using known methods, see for example \cite{shokrollahi2006raptor}, for analysis of raptor codes in BEC. This possibility will be explored in future works.  
%
The results also show that $f(.)$ increases as $N$ decreases. 


\subsection{Rapidly Reducing Network \label{subsec_rapid} }
It is instructive to consider a rapidly reducing network.
We set
$\lambda_l = 12$, $\lambda_e = 4$, $N=5000$, $\gamma = 98$, $\beta = 144$, $\alpha = 244$ and target decoding failure probability $\zeta = 10^{-12}$.
We assume that initially there are 5,000 nodes and 10,000 un-encoded blocks in the blockchain when implementing the proposed rateless coded blockchain approach. 
We assume that all blocks have a constant bit-size of 1 Mb. In order to clearly demonstrate the group size change, we assume that a maximum of 1 new enhanced block can be generated at each time epoch.  The network is run through 200 time epochs. 


Fig. \ref{sim_1} shows the number of network nodes and number of blockchain blocks at a given time epoch, and
the communication overhead needed by a new node for generating new parity check blocks or repairing intermediate blocks for previously encoded groups.  
The communication overhead is measured in the total number of bits that a new joining node needs.
The communication overhead of the traditional blockchain is the total bit-size of all blocks and it increases with the blockchain length. 
In contrast, the communication overhead of the rateless coded blockchain 
is less than 800 Mb when the total blockchain bit-size exceeds 38,000 Mb. 
The communication overhead alternatively increases and decreases because the blocks are consistently generated, and they are encoded only when the involved enhanced block is confirmed. 
%
It should be noted that the communication overhead increases with the number of encoded groups $|\mathcal{M}|$, and thus with the size $W$ of the blockchain. However, since $|\mathcal{M}|$ grows very slowly, the increment is not noticeable during the presented 200 time epochs.

In this experiment, 26 enhanced blocks are `mined'. In particular, the 3rd$\sim$13-th enhanced blocks are `mined' twice and the 1st$\sim$2nd enhanced blocks are `mined' thrice. The reason is that the number of nodes keeps decreasing, thus, the blocks in these groups need to be re-grouped and re-encoded to ensure the decoding failure probability meets the requirement. We show the time epochs when these blocks are `mined' and their $k_m$ values in Table \ref{tab_k_m}.
From  Table \ref{tab_k_m}, we see that the enhanced blocks are `mined' at a near constant interval of 12 time epochs after the initial stage at which there are a large pool of available blocks. 
Table \ref{tab_k_m} also shows that the enhanced blocks are re-generated or `re-mined' immediately after $\gamma$ time epochs and their values of $k_m$ decreases with the time epoch. This is because the number of nodes in the network is rapidly decreasing. 


\begin{table}[t]
    \centering
    \caption{Enhanced block generation and $k_m$ values. }
    \begin{tabular}{|c|c|c|c|c|}
    \hline
        seq.  & 1 & 2 & 3 $\sim$ 13 & 14 $\sim$ 26   \\ \hline
        1st `mined'  & 1 & 2 & 3 $\sim$ 95 & 102 $\sim$ 198   \\ \hline
        $k_m$ & 1910 & 1906 & 1904 $\sim$ 1613 & 1593 $\sim$ 1287\\ \hline
        2nd `mined' & 100 & 101 & 102 $\sim$ 194 & - \\ \hline
        $k_m$ & 1600 & 1595 & 1593 $\sim$ 1301 & -\\ \hline
        3rd `mined' & 199 & 200 & - & -  \\ \hline
        $k_m$ & 1284 & 1279 & - & - \\ \hline
    \end{tabular}
    \label{tab_k_m}
\end{table}

\begin{figure}[h]
\centering
   \includegraphics[scale = 0.45]{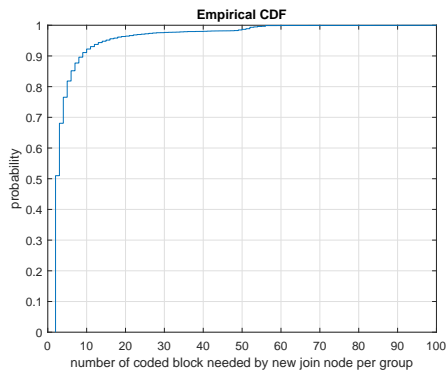}
  \caption{The CDF of the number of coded blocks that a new joining node needs to collect from other nodes in order to generate its parity or intermediate block for a given group. }
  \label{CDF_coded_block}
\end{figure}

Fig. \ref{CDF_coded_block} shows the CDF of the number of coded blocks, which can be systematic or non-systematic ones, that a new joining node needs to collect from other nodes in order to generate its own  parity or intermediate block for a given group.
%
We can see that with more than 90\% probability,  the new joining node only needs to collect 10 blocks.
This is significantly lower than other coded blockchains proposed in \cite{perard2018erasure}, \cite{wu2020distributed} and \cite{kadhe2019sef}. In these works a new node needs to collect more than $k_m$ blocks for group $G_m$, and as we saw $k_m$ is typically much larger than 10.  



It should be noted that this section simulates an extreme scenario where nodes consistently and rapidly leaving the network. In this case, existing nodes need to re-encode to ensure the blockchain is decodable, thus, increase the overall communication overhead. However, in practice, a typical blockchain network such as Bitcoin has a relative static number of nodes, e.g., the number of nodes in Bitcoin network is around 10,000 for many years, see in Section VII.C.
Moreover, in practice, we can employ some simple strategies to reduce the frequency of re-encoding. For instance, we can introduce an coefficient $r < 1$ in \eqref{equ_27} as $f(k_m, r| {\cal V}_{t_m} |; .) = \zeta$. Then `re-mining' an enhanced block will only occur when $|\mathcal{V}_{t'}| < r|\mathcal{V}_{t_{m}}|$.
In Bitcoin network, we can set $r=0.7$, which means `re-encoding' a group only occurs when the number of nodes is reduced by 30\%. 
%
As far as we know, this has not happened in Bitcoin history since 2015 \cite{bitcoin_number_2015}, and is not expected in future. 

\subsection{Performance in Bitcoin Network \label{sec_sim_bitcon} }

We evaluate our rateless coded blockchain using the real-world Bitcoin data obtained in \cite{bitcoin_data} from Nov. 18, 2018 to Nov. 13, 2020. We assume that our approach was implemented into the Bitcoin blockchain on Nov. 18, 2018. 
The initial and the latest blockchain lengths are 551,685 and 656,805 blocks, respectively. 
%
We assume that each time epoch is one day, i.e., 24 hours, hence we have $\beta = 144$, since one block is generated per 10 minutes on average \cite{nakamoto2008bitcoin}. 
We then assume that all blocks are fully loaded with bit-size of 1 Mb. 
From \cite{bitcoin_data} we obtain that nodes are leaving and joining the network as Possion processes with $\lambda_l = 42.18$ and $\lambda_e = 43.16$ respectively, and that the average number of nodes is 10,000.
We set $\alpha=144$ \cite{bitcoin_wiki}, $\zeta = 10^{-12}$ and $\gamma = 98$ (roughly 3 months). 
%
%
We omit the bit-size of block headers (latest about 50 Mb), which is negligible as compared with the total block data.
%
We define the {\em communication reduction coefficient} as the ratio of the communication overheads of the proposed rateless coded blockchain and the traditional blockchain.

\begin{figure*}[htbp]
  \centering
  \mbox{
    \subfigure[{ }\label{bitcoin_plot_ratio}]
    {\includegraphics[scale = 0.41]{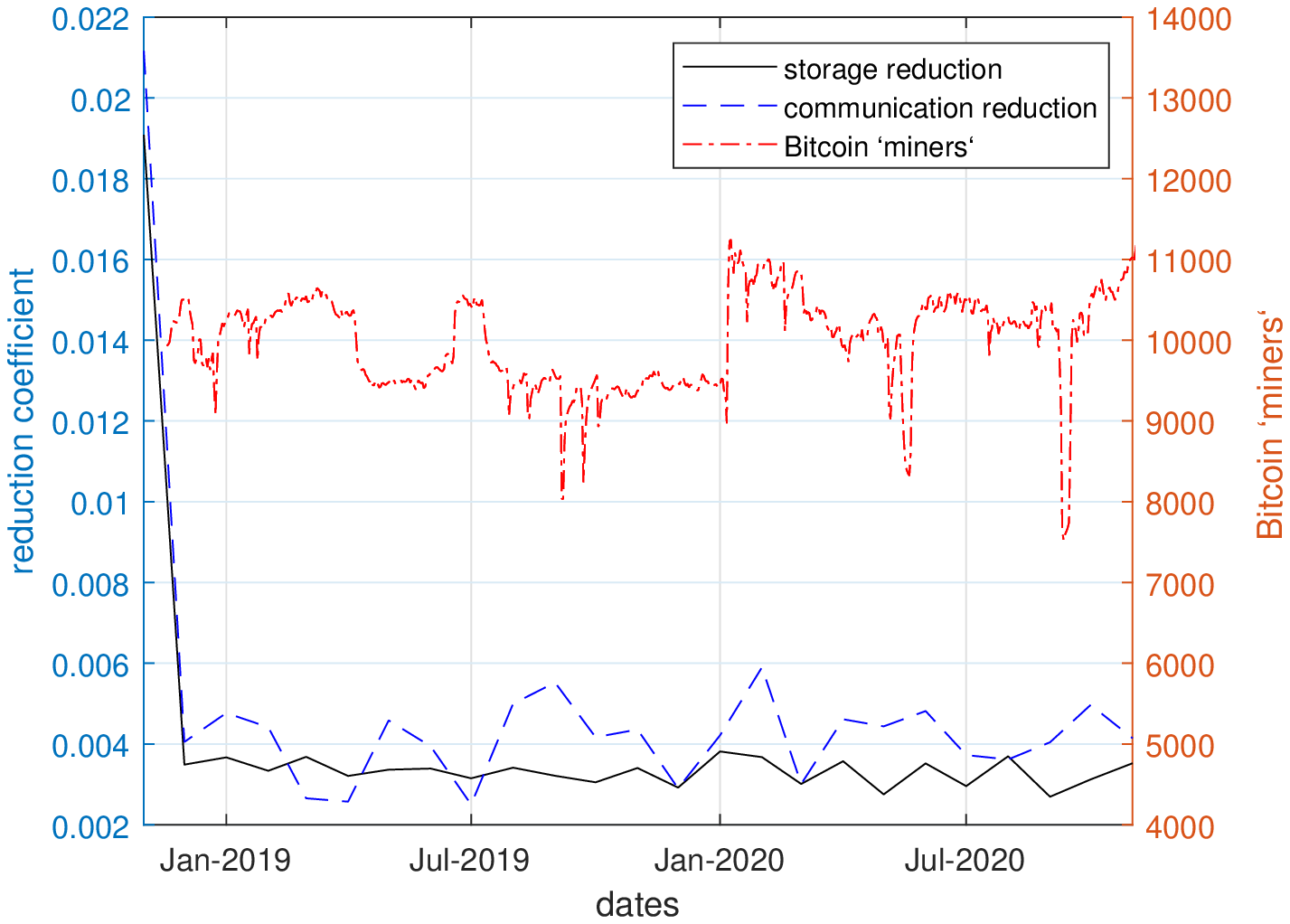}}\hspace*{-0.3em}
     \subfigure[{ }\label{bitcoin_plot_actual}]{\includegraphics[scale = 0.41]{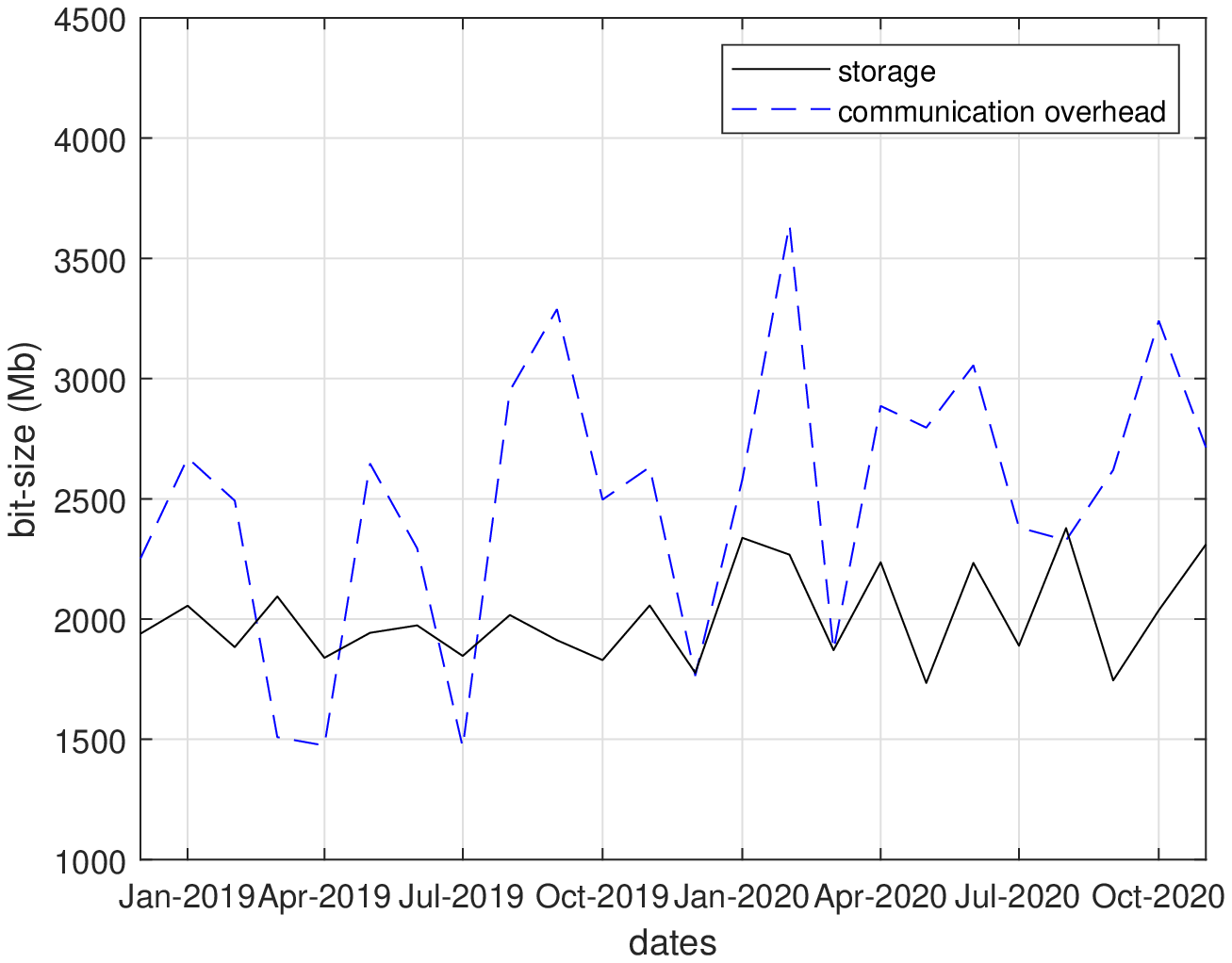}}\hspace*{-0.3em}
     \subfigure[{ }\label{bitcoin_cdf}]{\includegraphics[scale = 0.41]{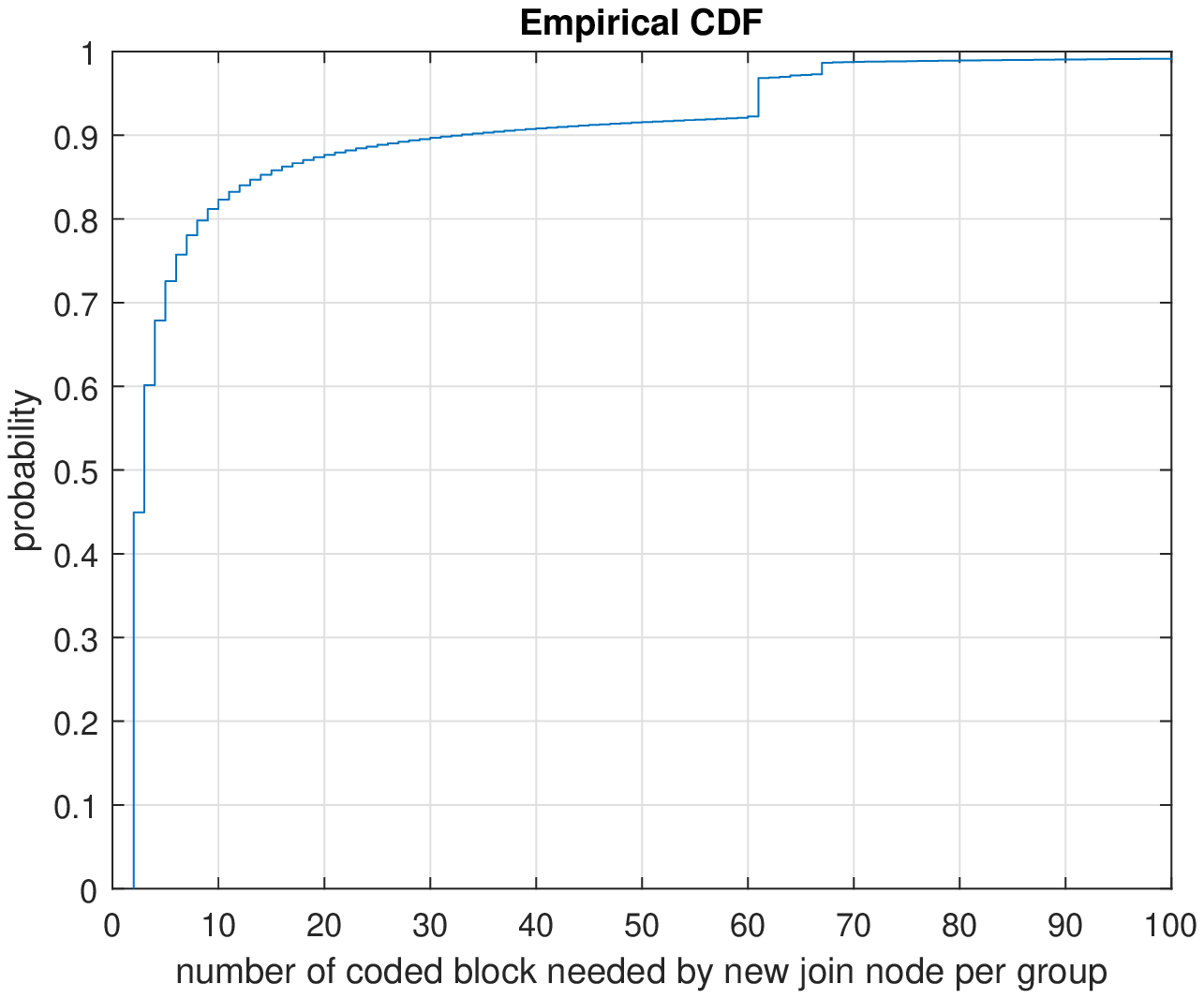}}
   }
  \caption{(a) Storage and communication reduction coefficients in real Bitcoin blockchain; (b) Storage and communication overhead bit-sizes in real Bitcoin blockchain; (c) The CDF of the number of coded blocks that a new joining node needs in real Bitcoin blockchain. }
  \label{ThreeFigures}
\end{figure*}



We show the month average storage reduction coefficient defined in \eqref{equ_storage} and the communication reduction coefficient in Fig. \ref{bitcoin_plot_ratio}. We see that both storage and communication reduction coefficients shrink to around 0.4\% and remain constants after Dec. 2018. 
This is because during the initial stage of the implementation, there are a very large number of original blocks that should be encoded into the coded blocks. Thus, the storage and communication reduction coefficients are both high at the beginning. 

In Fig. \ref{bitcoin_plot_actual}, we show the bit-size in Mb of the average storage at each node and the average communication overhead for each new joining node over a month. 
Recall that for the traditional replicated blockchain, both storage and communication overheads increase from 540 Gb to 640 Gb. So one can see the dramatic reduction by using the rateless coded blockchain.
Even though the blockchain size has increased by 100 Gb from Dec. 2018 to Nov. 2020, the average storage and communication overheads have increased by only 400 Mb and 700 Mb, respectively. The trends of storage and communication overheads in Fig. \ref{bitcoin_plot_actual} are similar to those in Fig. \ref{bitcoin_plot_ratio} because the Bitcoin blockchain length increases at a constant rate of 144 blocks per day.


Fig. \ref{bitcoin_cdf} shows the CDF of the number of coded blocks, either systematic or non-systematic ones, that a new node collects. 
With very high probability, i.e., 99\%, the new joining node only needs to obtain 70 coded blocks from other nodes per group, while on average there are  $10000 \times 0.337 = 3370$ blocks in each group. 
%
The number of blocks collected by a new node is slightly higher than that in Fig. \ref{CDF_coded_block} in Section \ref{subsec_rapid}. This is because there are more nodes in the Bitcoin network - 10,000 as opposed to 5,000 assumed in  Section \ref{subsec_rapid}. 
Hence the size of the groups and the expected generator degree in the Bitcoin network are both higher than those in Section \ref{subsec_rapid}. 
Although it is hard to observe in Fig. \ref{bitcoin_cdf}, we found that the probability that a new node has to decode a group $\mathcal{G}_m$, rather than directly encode or repair, is less than $6.3 \times 10^{-5}$. The decoding has relatively large complexity, so it is preferable to keep this probability low. 


In summary, if all Bitcoin blockchain nodes store the original blocks that are needed for verifying new transactions (around 5 GB \cite{bitcoin_core}) and encode all other old blocks using the proposed rateless coded blockchain, then each node needs a total storage space of around 7 GB. 
In this way, all old blocks are recoverable rather than simply deleted. In addition, when a new Bitcoin node joins the network, it only needs to download 8 GB of data.
%
These values are significantly lower than those of the current Bitcoin, which needs to download and store more than 300 Gb (600 Gb if all blocks are fully loaded).

\section{Conclusions \label{sec_conclusion}}

We have proposed a rateless coded blockchain to reduce the node storage requirement and the communication overhead for new joining nodes in practical dynamic IoT blockchain networks. It adjusts the coding parameters according to the network conditions to guarantee that all encoded blocks can be decoded whenever needed with very high probabilities. We have also presented novel enhanced blocks, which are embedded into the blockchain, to store the coding parameters, so that these parameters are immutable and can be shared among all nodes via the integrated and distributive features of the blockchain. Extensive simulations have been conducted using both synthetic data and real-world Bitcoin data. Results show that both the storage requirement at each node and the communication overhead of each new joining node are reduced by 99.6\% as compared with the traditional uncoded blockchain. 
A potential future work is to determine coding parameters of raptor code, such as pre-code rate and degree distribution, to have a lower decoding failure probability in the context of coded blockchain. 

\bibliographystyle{ieeetr}
\bibliography{ref}
\end{document}